\documentclass[journal]{IEEEtran}
\usepackage{amsmath}
\usepackage{amsfonts}
\usepackage{amsthm}
\usepackage{amssymb}
\usepackage{algcompatible}
\usepackage[linesnumbered,ruled]{algorithm2e}
\usepackage{cite}
\usepackage{dsfont}
\usepackage{pgfgantt}
\newtheorem{theorem}{Theorem}
\newtheorem{lemma}{Lemma}
\newtheorem{corollary}{Corollary}
\usepackage{float}
\usepackage{hyperref}
\usepackage{acronym}
\usepackage[caption=false,font=normalsize,labelfont=sf,textfont=sf]{subfig}
\graphicspath{{Schematics/}} 
\usepackage{xcolor}
\usepackage{hyperref}
\usepackage{subfig}
\usepackage{glossaries}
\usepackage{acronym}

\newacronym{ULA}{ULA}{uniform linear array}
\newacronym{UPA}{UPA}{uniform planar array}
\newacronym{TTD}{TTD}{true-time-delay}
\newacronym{HDB}{HDB}{Homomorphic Directional Beamforming}
\newacronym{OFDM}{OFDM}{orthogonal frequency-division multiplexing}
\newacronym{UE}{UE}{user equipment}
\newacronym{DTFT}{DTFT}{Discrete Time Fourier Transform}
\newacronym{ASE}{ASE}{Average Spectral Efficiency}

\title{Homomorphic Directional Beamforming with Analog True Time Delay Arrays \\
\thanks{This work was supported by the National Science Foundation awards 1955672 and 2224322.}}

\author{\IEEEauthorblockN{Ibrahim Pehlivan, and Danijela Cabric	}
		
\IEEEauthorblockA{\textit{Electrical and Computer Engineering Department,} \\
\textit{University of California, Los Angeles}\\
Emails: ipehlivan@ucla.edu, danijela@ee.ucla.edu }
}

\begin{document}

\maketitle
\begin{abstract}
 
Recently, \gls{TTD} arrays, also referred to as joint phase-time arrays (JPTA), have been investigated for low-cost frequency-dependent beamforming capabilities to enable various applications, including beam-squint correction, fast beam training, and serving multiple \gls{UE}s by frequency band to direction mapping, termed as split beampatterns. Several heuristics and optimization-based solutions have been proposed to determine \gls{TTD} array parameters settings. However, they have practical limitations due to either computationally demanding optimization procedures, requirements for extremely large memory look-up tables, or degradations due to the beam-squint effect. In this article, we propose a novel split-beampattern generation algorithm based on the observed homomorphism between \gls{TTD} array configuration matrices and corresponding beampatterns. First, we rigorously analyze the beampattern synthesis process and demonstrate the observed homomorphism and mathematical structure. Then, we propose the \gls{HDB} algorithm to approximate the desired split beampatterns by utilizing a generator beampattern dictionary that requires a dictionary size orders of magnitude lower than existing approaches without ignoring the beam squint. With extensive simulations, we show that the proposed algorithm can provide a practical implementation with low memory and low computational cost requirements. In addition, \gls{HDB} design provides close to uniform beamforming gains among \gls{UE}s in different subbands, enabling fairness in power allocation.
\end{abstract}
\section{Introduction}

\begin{figure}[t!]
\centering
\includegraphics[width=0.5\textwidth]{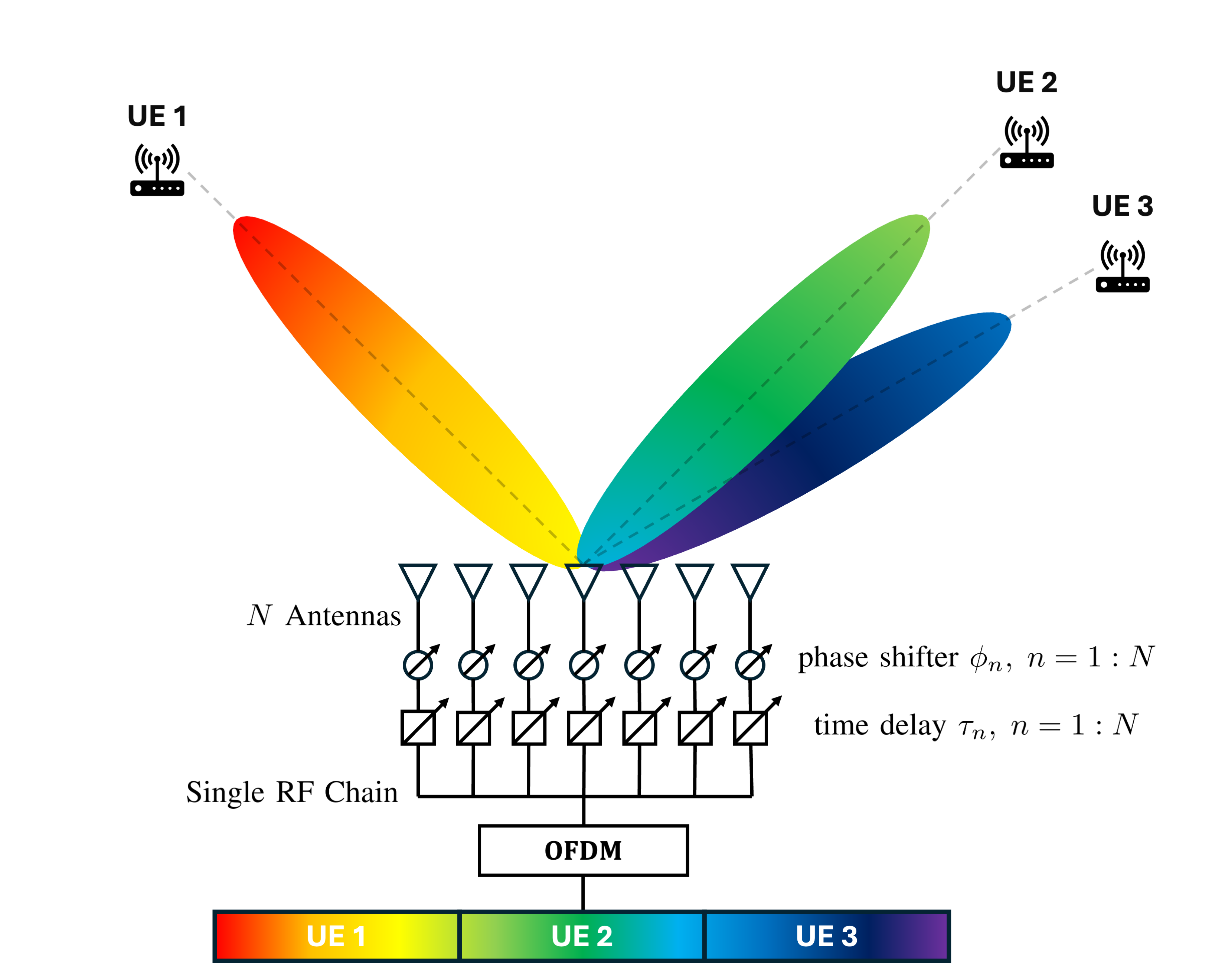}
\caption{Frequency-dependent beamforming with \gls{TTD} array: Split beampatterns that serve $3$ \gls{UE}s simultaneously with single RF chain by subband to \gls{UE} (direction) assignment.}
\label{fig:split_beam_system}
\end{figure}

Future wireless systems are expected to enable various wireless services ranging from virtual reality (VR) to ubiquitous connectivity, with stringent data rate, latency, and connectivity requirements \cite{Motivation_singh,samsung2025vision}. Enabling technologies are expected to utilize millimeter wave and sub-THz frequency bands to exploit the abundant bandwidth. However, these frequency bands exhibit high attenuation loss and require beamforming gain of the multi-antenna systems to enable viable communication links. Considering stringent power and cost limitations, analog antenna array architecture with frequency-independent phase shifters has been considered a predominant candidate in millimeter-wave systems \cite{25Years_array_Bjornson}. However, analog arrays can only create a single frequency-independent beam per given time slot, limiting the scheduling capabilities of the wireless systems. Furthermore, as the bandwidth and number of antennas increase, analog phase antenna arrays suffer from the frequency-dependent array response \cite{spatial_wideband}, referred to as beam squint. 

Recently, analog \gls{TTD} circuit elements have been proposed to enable analog frequency-dependent beamforming capabilities \cite{TTD_longbrake,TTD_rothman}, where introduced delay in the time domain provides the ability to create frequency-dependent phase shift in the frequency domain. This frequency-dependent phase-shifting capability has been utilized to eliminate the frequency dependency of the array response by re-aligning all frequencies towards the same direction \cite{delay_phase, TTD_longbrake, TTT_rainbow_hardware_WSU_II}. Instead of compensating for the beam squinting effect, \gls{TTD} elements can also increase the dispersion of the array response across different frequencies, creating so-called \textit{rainbow beampattern}. Rainbow beams have been proposed to reduce the beam training overhead by probing all possible \gls{UE} directions simultaneously for \gls{ULA} \cite{TTD_rainbow_Jans,TTD_rainbow_Han,TTD_rainbow_Veljko,TTT_rainbow_hardware_WSU,TTD_rainbow_jans_II,TTT_rainbow_hardware_WSU_II} and for \gls{UPA} \cite{3D_rainbow}. The array response dispersion is also utilized to facilitate low latency massive connectivity \cite{TTD_rainbow_link} and coverage extension with beam spreading \cite{TTD_Thzprism}.

In addition to dispersive beams, analog \gls{TTD} arrays have also been utilized to simultaneously serve multiple \gls{UE}s at different directions based on subband to direction multiplexing with \textit{split beampatterns}\cite{jpta}, as shown in Fig. \eqref{fig:split_beam_system}. The flexible spectrum utilization of split beams is currently being considered as a candidate paradigm for future 6G standards \cite{samsung2025vision} and has been shown to enhance uplink coverage \cite{JPTA_coverage}, improve edge user throughput \cite{JPTA_system} and enable low-latency operation \cite{mmflexible}. However, the aforementioned applications and increased degrees of freedom introduced by time delay elements bring additional algorithm design challenges that require careful analysis.

Several optimization-based beam-pattern generation algorithms have been proposed to enable flexible sub-carrier to direction assignment that is capable of beam-squint elimination, rainbow beamforming, and split beam realization. In \cite{jpta,mmflexible}, authors optimized the sum power of the \gls{OFDM} symbol for different subcarriers for \gls{ULA} arrays, proposing both exhaustive search-based and convex approximation-based solutions. \cite{mmflexible} proposes a closed-form solution for split beampatterns, where partitions of subcarriers are assigned to different directions while ignoring the beam-squint effect. The over-the-air (OTA) verifications of the aforementioned optimization-based algorithms are provided for $28$ GHz \cite{JPTA_demo} and sub-$6$ GHz \cite{mm_flexible_demo} systems. Furthermore, for \gls{UPA} arrays, optimization of the max-min power, in addition to sum power, is developed in \cite{3D_JPTA} to provide fairness among \gls{UE}s. However, optimization-based schemes either require prohibitive time complexity, memory \cite{jpta}, or ignore the beam-squint effect \cite{mmflexible}; making the online deployment for wideband systems impractical.

Several heuristic algorithms are proposed for computationally fast split beam generation; in \cite{modulo}, authors propose a closed-form configuration to serve $2$ \gls{UE}s by assigning each \gls{UE} half of the frequency band and \cite{structured} proposed a structured method to serve multiple \gls{UE}s that are uniformly distributed in direction domain. However, these closed-form methods focus on specific direction-sub-carrier assignments and have limitations for practical applications with diverse requirements on the number of users and their spatial directions.

In this article, we introduce a novel mathematical framework for computationally fast and memory-efficient split beampattern synthesis using analog TTD arrays. We first show the existence of a homomorphism between the array configuration matrix, i.e., time delay and phase shifter values per antenna and corresponding beampatterns. Utilizing this mathematical structure, we show that hard-to-approximate beampatterns, such as split beampatterns serving multiple \gls{UE}s, can be written in terms of simple-to-approximate generator beampatterns, allowing fast and efficient split beam synthesis in a divide-and-conquer manner. Our contributions are summarized as follows:\\

\begin{itemize}
    \item We rigorously analyze the mathematical structure of the \gls{TTD} beampattern synthesis and show the homomorphism between \gls{TTD} array configuration matrices and corresponding beampatterns.
    \item We propose a novel low-memory and low-complexity algorithm that approximates split beampatterns using a single generator beampattern dictionary. 
    \item We demonstrate the effectiveness of the proposed method with extensive simulations and provide a detailed comparison with the state-of-the-art algorithms.\\
\end{itemize}

The remainder of the paper is organized as follows: We introduce the system model and mathematical definitions in section \eqref{Section: System model}, define the optimization problem in section \eqref{Section: Problem Definition}, and analyze the mathematical structure of the beampattern synthesis operation in section \eqref{Section: Beampattern Groups and Homomorphism}. Then, we propose a split beampattern synthesis algorithm in section \eqref{Section: Split Beampattern Approximation} and present the extensive performance analysis in section \eqref{Section: Performance Analysis}. Finally, we discuss the properties of the proposed algorithm and future directions in section \eqref{Section: Discussion and Future Directions} and Section \eqref{Section: Conclusion} concludes the paper.

\section{System Model} \label{Section: System model}

We consider a wireless system comprising a base station with a single RF chain and $N$ antennas in \gls{ULA} configuration with half wavelength $\frac{\lambda_c}{2}$ antenna spacing at the carrier frequency $f_c$. Each antenna element employs a phase shifter with phase shifter value $\phi_n$ and a \gls{TTD} unit with a time delay value $t_n$ \cite{jpta,mmflexible}. $M$ subcarrier \gls{OFDM} system with bandwidth $BW$ is employed for the data communication where each subcarrier $m=1:M$ has the frequency $f_m=f_c+m(BW/M)-(BW/2)$. 

The resulting frequency dependent precoding vector $\mathbf{v}_{m}\in \mathbb{C}^{N \times 1}$ for the sub-carrier $m \in 1:M$ and for time delay $t_n$ and phase shift $\phi_n$ per antenna $n\in 1:N$ can be defined as follows:

\begin{equation}
    [\mathbf{v}_{m}]_{n}=\dfrac{1}{\sqrt{N}}e^{j(-2\pi f_m t_n+\phi_n)} \in \mathbb{C}^N,\;m\in1:M
\end{equation}

and the precoding matrix is defined as:

\begin{equation}
\mathbf{V}= \begin{bmatrix}
    \mathbf{v}_{1} &\dots & \mathbf{v}_{M}
\end{bmatrix} \in \mathbb{C}^{N \times M} \label{Definition: Precoder matrix}
\end{equation}

where the corresponding array response for a given $t_n$ and phase shift $\phi_n$, per antenna $n\in 1:N$ and for the sub-carrier $m \in 1:M$ is given as\cite{jpta}:

\begin{equation}
    \begin{aligned}
        &P_{m}(\Psi)=\sum_{n=0}^{N-1} \frac{1}{\sqrt{N}}e^{j(-2\pi f_m t_n+\phi_n)} e^{-jn\pi\Psi \frac{f_m}{2f_c}}\nonumber\\
        &P_{m}(\Psi)=\sum_{n=0}^{N-1} \frac{1}{\sqrt{N}}e^{j(-2\pi f_m t_n+\phi_n)} e^{-jn\Omega_{m}}\nonumber\\
        &=\mathcal{F}\{\dfrac{1}{\sqrt{N}}e^{j(-2\pi f_m t_n+\phi_n)}\}_{(\Omega_m)}\nonumber\\
        &:=\mathcal{F}_m\{\dfrac{1}{\sqrt{N}}e^{j(-2\pi f_m t_n+\phi_n)}\} \label{Def: F_m}
    \end{aligned}
\end{equation}

where $\Psi=\sin(\theta)$ and $\mathcal{F}_m$ is the \gls{DTFT} with the frequency variable $\Omega_m=\pi\Psi f_m/(2f_c)$, i.e scaled version of \gls{DTFT} with the frequency variable $\pi \Psi$.

\begin{figure}[!t]
\centering
  \includegraphics[width=0.5\textwidth]{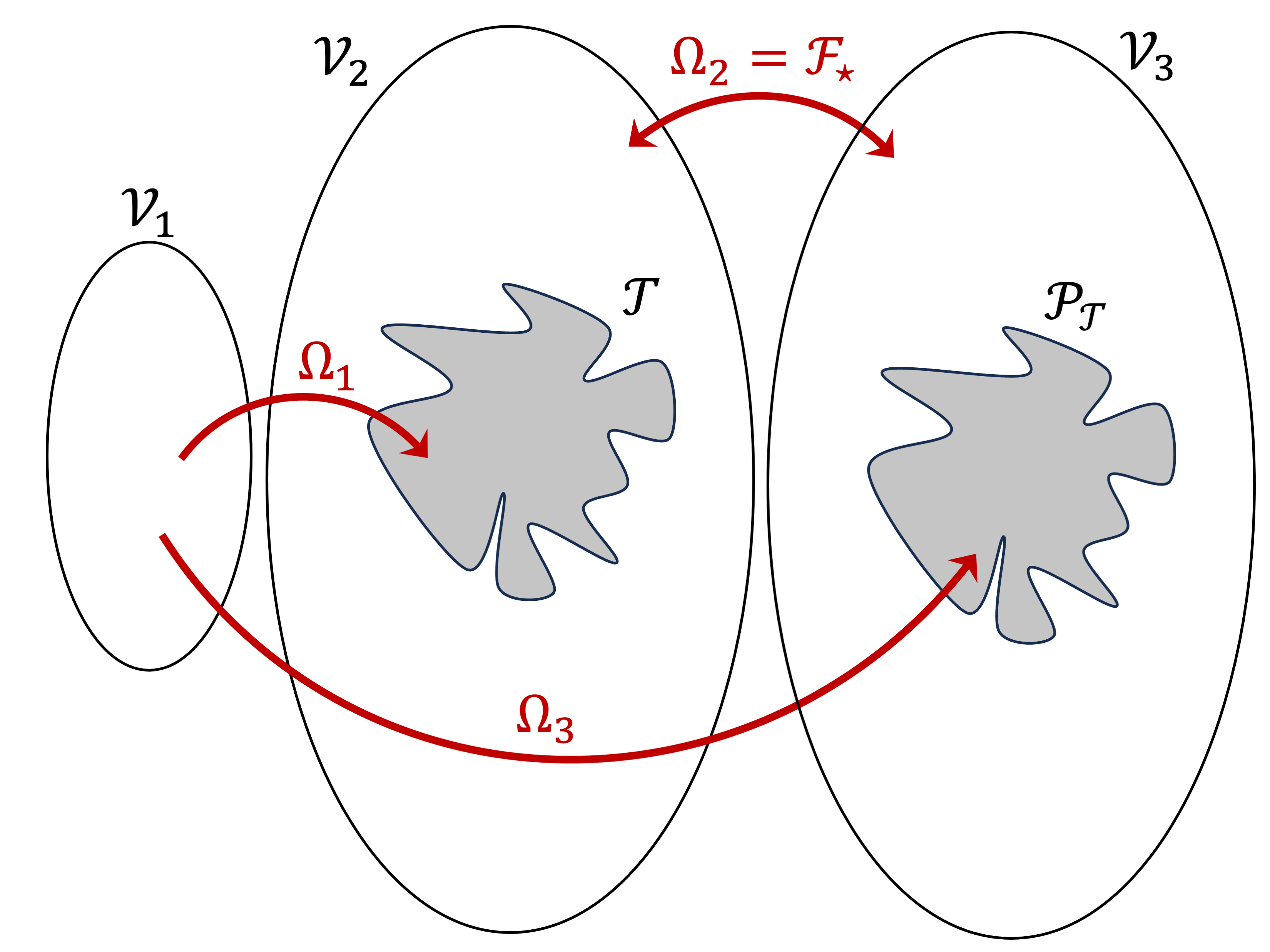}
  \caption{Beampattern generation process: sets and corresponding mappings. $\mathcal{V}_1$ is the set of array configuration matrices, $\mathcal{V}_2$  is the set of beamforming vectors and $\mathcal{V}_3$ is the set of beampatterns.}
  \label{fig:Mapping_Func}
\end{figure}
\subsection{TTD Array and Beampattern Set Mappings}

 We start by defining the necessary mathematical tools to analyze the \gls{TTD} beamforming process. Denote the frequency vector $\mathbf{f}\in \mathbb{R}^{M}$ as $[\mathbf{f}]_{(m)}=f_m, m\in 1:M$, time delay and phase shifter vectors as $[\mathbf{t}]_{(n)}=t_n,\;[\boldsymbol{\phi}]_{(n)}=\phi_n, n\in 1:N$ respectively and array configuration matrix as $\boldsymbol{\Phi}=\begin{bmatrix}
     \mathbf{t} & \boldsymbol{\phi}
 \end{bmatrix}\in \mathbb{C}^{N \times 2}$ and define the following operators:

  \begin{subequations}
  \begin{align}
      &\Omega_1(\boldsymbol{\Phi}): \mathcal{V}_1 \to \mathcal{V}_2:= \dfrac{1}{\sqrt{N}}e^{j(-2\pi \mathbf{t}^H \mathbf{f}+\boldsymbol{\phi})}:=\mathbf{V_{\Phi}}\label{Def: Omega1}\\
      &\Omega_2(\mathbf{V})=\mathcal{F}_{\star}(\mathbf{V}):= \mathcal{V}_2\to \mathcal{V}_3\;\nonumber\\ 
      &\quad \qquad =\mathcal{F}_m\{\mathbf{V}(:,m)\}=P_{m}(\Psi):=[\mathbf{P}]_{(\Psi,m)},\; \forall m  \\     
      &\Omega_3(\boldsymbol{\Phi}): \mathcal{V}_1\to \mathcal{V}_3 :=\Omega_2(\Omega_1 (\boldsymbol{\Phi}))\label{Def: Omega3}
  \end{align}
\end{subequations}

Where $\mathcal{V}_1=<\mathbb{R}^{N \times 2},+,.>$ is array configuration matrix vector space, $\mathcal{V}_2=<\mathbb{C}^{N \times M},+,.>$ is beamforming matrix vector space, and $\mathcal{V}_3=<\mathbb{C}^{[-1,1]},+,.>$ is beampattern \cite{mmflexible} vector space with respect to standard vector addition and multiplication. $\Omega_{2}:=\mathcal{F}_{\star}$ is the column-wise, per subcarrier $m$, $\mathcal{F}_{m}$ operation defined in the equation \eqref{Def: F_m}, $\Omega_1$ is the \gls{TTD} array generating function and $\Omega_3$ is the beampattern generation function from array configuration matrix. $\mathcal{T}:=\Omega_{1}(\mathcal{V}_1)\subset \mathcal{V}_2 $ is the set of \gls{TTD} beamformers i.e. image of the $\mathcal{V}_1$ over the operator $\Omega_1$. $\mathcal{P}_{\mathcal{T}}:=\Omega_{2}(\mathcal{T})$ is the set of \gls{TTD} beampatterns generated by \gls{TTD} Arrays $\mathcal{T}$. Observe that $\mathcal{T}$ and $\mathcal{P}_{\mathcal{T}}$ are non-convex from the properties of the $\Omega_1$ and the linearity of the $\Omega_2$ operators. The sets and mappings are illustrated in the figure \eqref{fig:Mapping_Func}.

\section{Problem Definition and State-of-the-art} \label{Section: Problem Definition}

 The goal of the paper is to find the array configuration matrix $\boldsymbol{\Phi}=\begin{bmatrix}
     \mathbf{t} & \boldsymbol{\phi}
 \end{bmatrix}\in \mathbb{C}^{N \times 2}$ to generate the beampattern $\mathbf{P_{\boldsymbol{\Phi}}}$ that minimizes the following multi-objective optimization problem \cite{convexopt_boyd}: 

 \begin{equation}\label{eq: obj_fun_1}
\min_{\boldsymbol{\Phi}\in \mathbb{R}^{N \times 2}} \oint \|[\mathbf{P_{\Phi}}]_{(\Psi,m)}-[\mathbf{P}_{\textit{target}}]_{(\Psi,m)}\|^2 d\Psi, \; \forall m
\end{equation}

 where $\mathbf{P}_{\textit{target}}$ is the desired beampattern. This problem is non-convex and intractable as the feasible set of \gls{TTD} beampatterns $\mathcal{P}_{\mathcal{T}}$ is non-convex. Observe that the problem \eqref{eq: obj_fun_1} can be optimally solved for digital arrays (for relaxed problem) as digital arrays can design beamforming vectors of subcarriers independently. We continue with sum-objective based near-optimal solution strategies\cite{mmflexible,jpta}:

 \subsection{Sum-objective optimization over TTD Beampatterns}
 

 The first sum-objective based solution strategy for the problem \eqref{eq: obj_fun_1} optimizes the following optimization problem over vector space $\mathcal{V}_3$ \cite{mmflexible}: 

 \begin{equation}\label{eq: obj_fun_2}
\min_{\boldsymbol{\Phi}\in \mathbb{R}^{N \times 2}} \sum_{m=1}^{M} \oint \|[\mathbf{P_{\Phi}}]_{(\Psi,m)}-[\mathbf{P}_{\textit{target}}]_{(\Psi,m)}\|^2 d\Psi
\end{equation}

 Although the multi-objective problem is reduced to a single objective problem, the problem \eqref{eq: obj_fun_2} is still non-convex as the feasible set remains the same, and obtaining an optimal solution is still challenging. An exhaustive-search-based solution for the optimization problem (\ref{eq: obj_fun_2}) is proposed over the relaxed feasible set, and a low-complexity approximate solution is proposed for the approximate $\Omega_{2}$ \cite{mmflexible}.  However, the proposed algorithms ignore the beam-squint effect \cite{spatial_wideband} i.e. assumes $\mathcal{F}_m\{.\}$, defined in equation \eqref{Def: F_m} is $\mathcal{F}_m\{.\}\approx\mathcal{F}\{.\}$.
 
 
  \subsection{Weighted sum-objective optimization over TTD Beamformers}

  Alternatively, \cite{jpta} solves the following optimization problem over the vector space $\mathcal{V}_2$:

 \begin{equation}
    \min_{\boldsymbol{\Phi}\in \mathbb{R}^{N \times 2}} \sum_{m=1}^{M}\sum_{n=1}^{N} |[\mathbf{V_{\Phi}}]_{(n,m)}-[\mathbf{V}_{\textit{target}}]_{(n,m)}|^2 \label{eq: obj_fun_3}
 \end{equation}

The problem \eqref{eq: obj_fun_3} is equivalent to the weighted sum-objective solution of the \eqref{eq: obj_fun_1} by generalized Parseval's equality \cite{vetterli}, where weights are introduced due to the changing signal period per subcarrier $m$ from scaled \gls{DTFT} definition in equation \eqref{Def: F_m}. The problem \eqref{eq: obj_fun_3} is non-convex and NP-hard \cite{jpta}; therefore, alternating minimization-based solution methodology is proposed in \cite{jpta} by either employing exhaustive search or by linearizing the \gls{TTD} generation function $\Omega_{2}$.


 
The aforementioned sum-optimization-based solution methods either require exhaustive search or involve some way of convex relaxation on feasible sets $\mathcal{T}$ or $\mathcal{P}_{\mathcal{T}}$. We want to avoid this route; Instead, we show that the set of the vector space $\mathcal{V}_3$ forms a group structure with an appropriate addition operation, and the beampattern synthesis operation in \eqref{Def: Omega3} forms a homomorphism. Hence, we can approximate and generate complicated beampatterns from known and simpler beampatterns to simplify the beampattern synthesis process, following the general wisdom of signal processing. We continue with a rigorous understanding of the beampattern generation process.


\section{Homomorphism over TTD Beampatterns} \label{Section: Beampattern Groups and Homomorphism}

We start by showing the function $\Omega_3$ defines a homomorphism \cite{abstract_algebra_fraleigh} between $<\mathcal{V}_1,+>$ and $<\mathcal{P}_{\mathcal{T}}\subset \mathcal{P},\star>$ where $\star$ is the scaled column-wise circular convolution operation defined as:

\begin{subequations}
\begin{align}
& \text{Let }\mathbf{P^{'},P^{''}}\in\mathcal{P},\nonumber\\
&\star: \mathcal{P} \times \mathcal{P} \to \mathcal{P}
:= \nonumber \\
&\qquad \quad [\mathbf{P^{'}} \star \mathbf{P^{''}}]_{(:,m)}=\sqrt{N}\left([\mathbf{P^{'}}]_{(:,m)} \circledast [\mathbf{P^{''}}]_{(:,m)}\right), \forall \; m \label{Def: star_addition}
\end{align}
\end{subequations}

where the operation $\circledast$ is the circular convolution. 

\begin{theorem}
   $\Omega_3$ is a homomorphism \cite{abstract_algebra_fraleigh} between $<\mathcal{V}_1,+>$ and $<\mathcal{P}_{\mathcal{T}},\star>$ :
   \begin{subequations}
    \begin{align}
        &\forall \; \boldsymbol{\Phi}_{1},\boldsymbol{\Phi}_{2} \in \mathcal{V}_1 = \mathbb{R}^{N \times 2} \nonumber\\
    &\quad \Omega_3\left(\boldsymbol{\Phi}_{1}+\boldsymbol{\Phi}_{2}\right)=\Omega_3\left(\boldsymbol{\Phi}_{1}\right)\star\Omega_3\left(\boldsymbol{\Phi}_{2}\right), \; \nonumber 
    \end{align}
    \end{subequations} 
    \label{theorem: homomorphism}
\end{theorem}

\begin{proof}
    \begin{subequations}
    \begin{align}
    &\text{Let, }\; \boldsymbol{\Phi_1}=\begin{bmatrix}
     \mathbf{t_1} & \boldsymbol{\phi_1}
 \end{bmatrix},\boldsymbol{\Phi_2}=\begin{bmatrix}
     \mathbf{t_2} & \boldsymbol{\phi_2}
 \end{bmatrix} \in \mathcal{V}_1 = \mathbb{R}^{N \times 2}, \\ &[\Omega_3\left(\boldsymbol{\Phi_1}+\boldsymbol{\Phi_2}\right)]_{(:,m)}  =\mathcal{F}_m\{\dfrac{1}{\sqrt{N}}e^{j(-2\pi f_m  (\mathbf{t_{1}+t_{2}})+(\boldsymbol{\phi_{1}+\phi_{2}} )}\}, \nonumber\\
 &\qquad \qquad \qquad \qquad \qquad \qquad \qquad \qquad \qquad \qquad \forall m \label{proof: homomorphism - 1}\\
    &=\sqrt{N}\mathcal{F}_m\{\dfrac{1}{\sqrt{N}}e^{-j(2\pi f_m \mathbf{t_{1}+\phi_{1}})}\dfrac{1}{\sqrt{N}}e^{-j(2\pi f_m \mathbf{t_{2}}+\boldsymbol{\phi_{2}})}\}\label{proof: homomorphism - 2}\\
    &=\sqrt{N}(\mathcal{F}_m\{\dfrac{1}{\sqrt{N}}e^{-j(2\pi f_m \mathbf{t_{1}}+\boldsymbol{\phi_{1}})}\}\nonumber\\
    &\qquad \qquad \qquad \qquad \qquad  \circledast \mathcal{F}_m\{\dfrac{1}{\sqrt{N}}e^{-j(2\pi f_m \mathbf{t_{2}}+\boldsymbol{\phi_{2}})}\} )\label{proof: homomorphism - 3} \\ 
    &=\sqrt{N}\left([\Omega_3\left(\boldsymbol{\Phi_1}\right)]_{(:,m)} \circledast [\Omega_3\left(\boldsymbol{\Phi_2}\right)]_{(:,m)}\right)\label{proof: homomorphism - 4}\\
    &\iff \Omega_3\left(\boldsymbol{\Phi_1}+\boldsymbol{\Phi_2}\right)=\Omega_3\left(\boldsymbol{\Phi_1}\right) \star \Omega_3\left(\boldsymbol{\Phi_2}\right) \label{proof: homomorphism - 5}
    \end{align}
    \end{subequations} 
\end{proof}

Where (\ref{proof: homomorphism - 1},\ref{proof: homomorphism - 4}) is from the definition (\ref{Def: Omega3}), (\ref{proof: homomorphism - 2}) is from properties of complex numbers, (\ref{proof: homomorphism - 3}) is from convolution property of discrete Fourier transform and (\ref{proof: homomorphism - 5}) is from the definition of $\star$ from (\ref{Def: star_addition}). 

\begin{corollary}
    \begin{equation}
        \begin{aligned}
                  &\forall \boldsymbol{\Phi}_{1},\boldsymbol{\Phi}_{2} \in \mathcal{V}_1 = \mathbb{R}^{N \times 2},\; \nonumber \mathbf{P}_{\boldsymbol{\Phi}_{1}+\boldsymbol{\Phi}_{2}}=\mathbf{P}_{\boldsymbol{\Phi}_{1}} \star \mathbf{P}_{\boldsymbol{\Phi}_{2}}   
        \end{aligned}
    \end{equation}
    \label{corollary: homomorphism}
\end{corollary}
\begin{proof}
    Proof follows from Theorem \eqref{theorem: homomorphism} and $\Omega_3$ definition in \eqref{Def: Omega3}.
\end{proof}

Convolution and its algebraic properties have been previously discussed under generalized linearity for homomorphic filtering in \cite{schafer_echo}, where homomorphic transformation is intentionally introduced for domain transformation (followed by the inverse transformation). For \gls{TTD} arrays, however, we observed and showed that homomorphism is the inherent property of the beampattern synthesis process.

$<\mathcal{P}_{\mathcal{T}},\star>$ defines a group structure with respect to $\star$ operator\cite{schafer_echo}. However, not every beampattern $\mathbf{P}\in \mathcal{P}$ has an inverse with respect to the $\star$ operator; hence, we cannot extend the group structure to the $<\mathcal{P},\star>$. To ensure invertibility and group structure, we restrict the set of beampatterns $\mathcal{P}$ to $\mathcal{P}_{nz}$ where the patterns are generated by the beamforming vectors with non-zero gain for every antenna $n\in 1:N$ and for every subcarrier $m \in 1:M$. This is true for most of our desired applications, including \gls{TTD} analog arrays, and we denote $\mathcal{P}_{nz}=\mathcal{P}$ for the rest of the paper for conciseness. We continue with demonstrating how a \gls{TTD} beampattern can be synthesized by utilizing the observed homomorphism. 

\subsection{Homomorphic TTD Beampattern Synthesis} \label{homom_section}

Let us consider a \gls{TTD} beampattern $\mathbf{P_{\Phi'}} \in \mathcal{P}_{\mathcal{T}}$ and assume that it can be represented with $G$ \gls{TTD} beampatterns over $\star$ operation such as:

\begin{equation}
\begin{aligned}
    \mathbf{P_{\boldsymbol{\Phi}'}}=\mathbf{P_{\boldsymbol{\Phi_1}}} \star \dots \star \mathbf{P_{\boldsymbol{\Phi_{G}}}} \label{eqn: star representation}
\end{aligned}
\end{equation}

Then, the Corollary \eqref{corollary: homomorphism} implies that we can write the corresponding array configuration matrices as:

\begin{equation}
\begin{aligned}
    \boldsymbol{\Phi}'= \sum_{g=1}^{G} \boldsymbol{\Phi_g} \label{eqn: config matrix represnt.}
\end{aligned}
\end{equation}

and can directly find time delay $\mathbf{t'}$ and phase shifter values $\boldsymbol{\phi'}$ as:
\begin{equation}
\begin{aligned}
    \mathbf{t'}=\sum_{l=1}^{G}\mathbf{t_g},\quad \boldsymbol{\phi'}=\sum_{l=1}^{G}\boldsymbol{\phi_g} \label{eqn: time and phase representation}
\end{aligned}
\end{equation}

Therefore, finding a representation of a \gls{TTD} beampattern over $\star$ operation is sufficient to determine its array configuration matrix. Although a representation of a beampattern $\mathbf{P}\in \mathcal{P}$ that is not a \gls{TTD} beampattern, i.e., $\mathbf{P}\notin \mathcal{P}_{\mathcal{T}}$ can also be found from the group structure of $<\mathcal{P},\star>$; utilization of the homomorphism and equation \eqref{eqn: time and phase representation} requires the approximation of the $\mathbf{P}$ or its representation by \gls{TTD} beampatterns in $\mathcal{P}_{\mathcal{T}}$.

\subsection{Beampattern Approximation: Divide-and-Conquer} \label{Section: Divide and conquer}

We denote the beampatterns that represent $\mathbf{P}$ as \textit{generator} beampatterns. Representing a beampattern $\mathbf{P}\in<\mathcal{P},\star>$ over $\star$ operation resembles the basis representation of a vector. Hence, instead of working on a hard-to-approximate beampattern, our intuitive approach is to work on corresponding easy-to-approximate \textit{generator} beampatterns and utilize the equations \eqref{eqn: config matrix represnt.},\eqref{eqn: time and phase representation} to synthesize the approximated \gls{TTD} beampattern, in a divide-and-conquer manner. Unfortunately, representing an arbitrary beampattern in $<\mathcal{P},\star>$ and determining \textit{generator} beampatterns is not a trivial task and requires more advanced mathematical tools beyond the utilized group structure. We leave this task to future research and instead focus on beampatterns that can be represented by easy-to-approximate \textit{generator} beampatterns such as split beampatterns.


\section{Split Beampattern Approximation} \label{Section: Split Beampattern Approximation}

\begin{figure}[!t]
\includegraphics[width=0.5\textwidth]{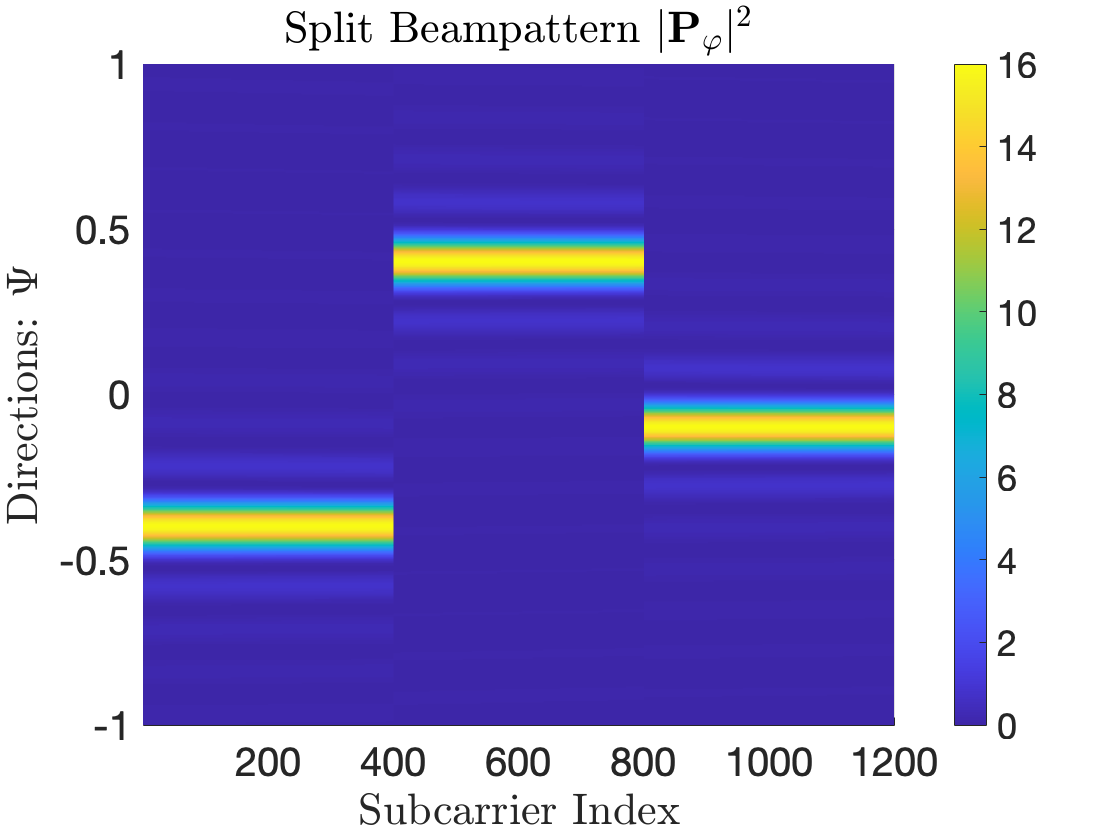}
\caption{Gain of a split beampattern with $G=3$ subbands where each subband is assigned to a direction with the subband to direction mapping vector: $\boldsymbol{\varphi}=\begin{bmatrix}
    -0.4 & 0.4 & -0.1
\end{bmatrix}^{H}$} .
\label{Fig: split beampattern example}
\end{figure}

We continue with the definition of split beampatterns and show that split beampatterns can be represented with easy-to-approximate \textit{generator} beampatterns. 

\subsection{Split Beampattern Definition}

A split beampattern with $G$ subbands $\mathbf{P}(\boldsymbol{\varphi}) \in \mathcal{P} \subset \mathcal{P}_{\mathcal{V}_2}$, assigns directions to subbands according to subband to direction mapping vector $\boldsymbol{\varphi}\in[1,-1]^{G\times1}$. Therefore, $\mathbf{P}(\boldsymbol{\varphi}) \in \mathcal{P} \subset \mathcal{P}_{\mathcal{V}_2}$ is defined to be generated by the following far-field beamforming matrix $\mathbf{V}(\boldsymbol{\varphi}) \in \mathcal{V}_2$ \cite{jpta}:

\begin{equation}
\begin{aligned}
    & [\mathbf{V}(\boldsymbol{\varphi})]_{(:,m)}=\dfrac{1}{\sqrt{N}}e^{j\pi n [\psi]_m \frac{f_m}{f_c} } \in \mathbb{C}^{N \times 1},\;m\in1:M\\
     &\boldsymbol{\psi}=\begin{bmatrix}
        \smash{\underbrace{\begin{matrix}
            \varphi_1 &\dots & \varphi_1 
        \end{matrix}}_{M'}} &  \dots & \smash{\underbrace{\begin{matrix}
            \varphi_B &\dots & \varphi_B 
        \end{matrix}}_{M'}} 
    \end{bmatrix}^{H} 
\end{aligned}
\end{equation}\\

where the subcarriers are
partitioned into $G$ subbands each containing $M'=\frac{M}{G}$ subcarriers and subcarrier $m$ is directed towards the direction $[\boldsymbol{\psi}]_m$. Then, the corresponding beampattern $\mathbf{P}(\boldsymbol{\varphi}) \in \mathcal{P}_{\mathcal{V}_2}$ is obtained as follows:

\begin{subequations}

\begin{align}
       &[\mathbf{P}(\boldsymbol{\varphi})]_{(:,m)}=[\Omega_2(\mathbf{V}(\boldsymbol{\varphi}))]_{(:,m)}, \nonumber\\
       &\qquad \qquad=\dfrac{1}{\sqrt{N}}\Xi_{N,m}(\Psi-\psi_m),\; \forall m\nonumber\\
    \label{Def: dirichlet beampattern}
\end{align}

\end{subequations}

where $\Xi_{N,m}$ is the Dirichlet sinc function for $N$ antennas and $f_m$ \cite{delay_phase} and a split beampattern example is given in figure \eqref{Fig: split beampattern example}. The following Lemma (\ref{Lemma: addition of angles}) shows an important property of the split beampatterns, i.e., when the $\star$ operation is applied to two split beampatterns, their corresponding direction mapping vectors are added:

\begin{lemma} 
    Let $\mathbf{P}(\boldsymbol{\varphi}_1),\mathbf{P}(\boldsymbol{\varphi}_2) \in \mathcal{P}_{\mathcal{V}_2}$, then $$\mathbf{P}(\boldsymbol{\varphi}_1+\boldsymbol{\varphi}_2)=\mathbf{P}(\boldsymbol{\varphi}_1) \star \mathbf{P}(\boldsymbol{\varphi}_2)$$ 
    \label{Lemma: addition of angles}
\end{lemma}
\begin{proof}
    \begin{subequations}
    \begin{align}
        &\text{Let, }\mathbf{P}(\boldsymbol{\varphi}_1),\mathbf{P}(\boldsymbol{\varphi}_2) \in \mathcal{P}_{\mathcal{V}_2} \nonumber,\\ 
        &[\mathbf{P}(\boldsymbol{\psi}_1) \star \mathbf{P}(\boldsymbol{\varphi}_2)]_{(:,m)}\nonumber\\
        &\quad =\sqrt{N}(\Xi_{N,m}(\Psi-[\boldsymbol{\psi_1}]_m)\circledast\Xi_{N,m}(\Psi-[\boldsymbol{\psi_2}]_m))\label{proof: lemma addition of angles - 1}\\
         &\qquad= \sqrt{N}\Xi_{N,m}(\Psi-([\boldsymbol{\psi_1}]_m+[\boldsymbol{\psi_2}]_m))\label{proof: lemma addition of angles - 2}\\
         & \qquad \qquad \qquad \qquad \qquad \qquad  =[\mathbf{P}(\boldsymbol{\varphi}_1+\boldsymbol{\varphi}_2)]_{(:,m)}, \; \forall m  \label{proof: lemma addition of angles - 3}
    \end{align}
    \end{subequations} 
\end{proof}

where equations \eqref{proof: lemma addition of angles - 1},\eqref{proof: lemma addition of angles - 3} follow from the definition in equation \eqref{Def: dirichlet beampattern}, and equation \eqref{proof: lemma addition of angles - 2} follows from the properties of the Dirichlet sinc function. The Lemma \eqref{Lemma: addition of angles} is the main tool for obtaining generator beampatterns that represent the split beampatterns. We continue with the construction of the generator beampatterns of split beampatterns. Then, we discuss the properties of the generators that make them easy to approximate.

\begin{figure*}[!t]
\centering
\subfloat[]{\includegraphics[width=0.4\textwidth]{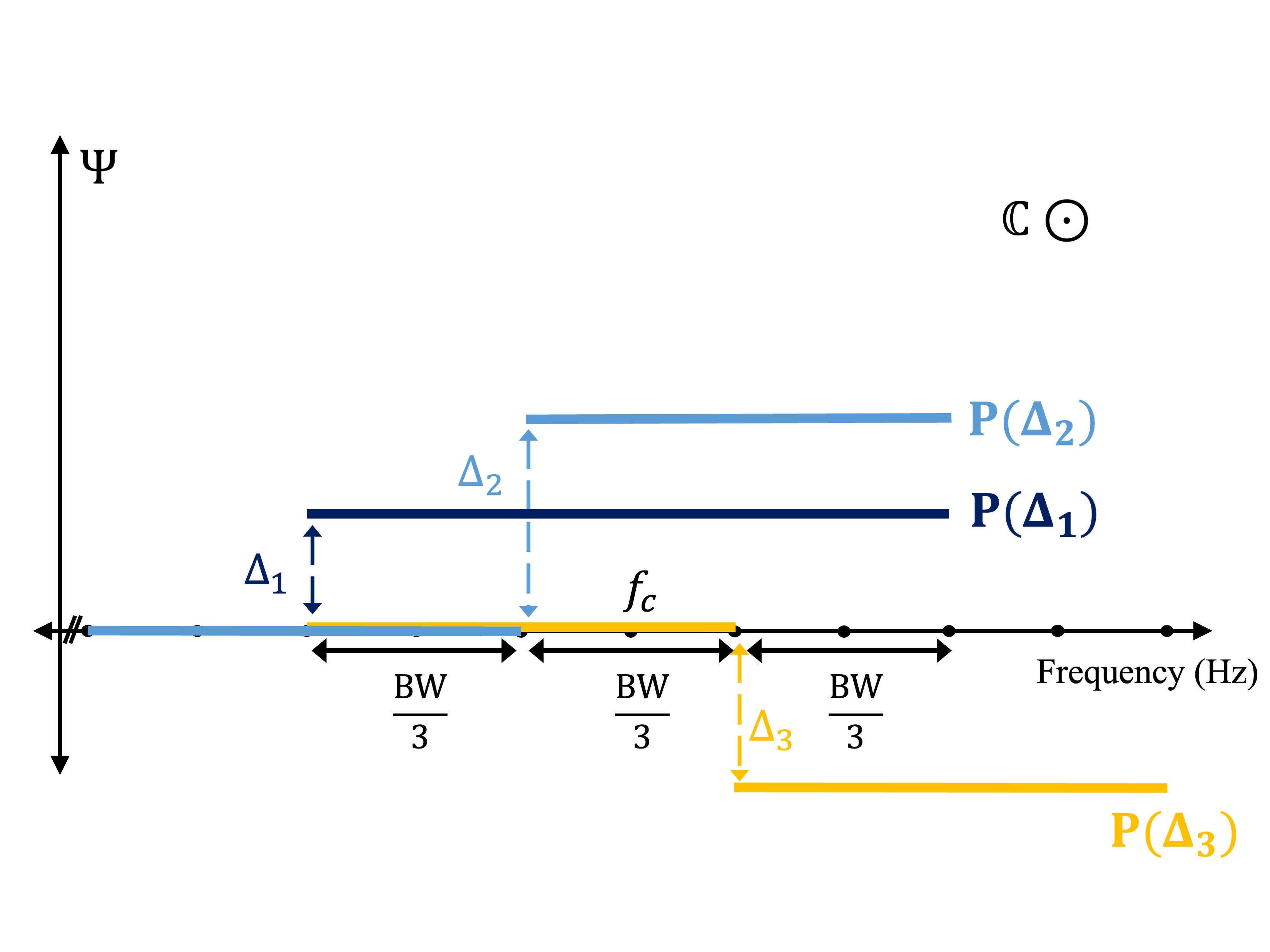}%
\label{Fig: 3 basis vectors}}
\hfil
\subfloat[]{\includegraphics[width=0.4\textwidth]{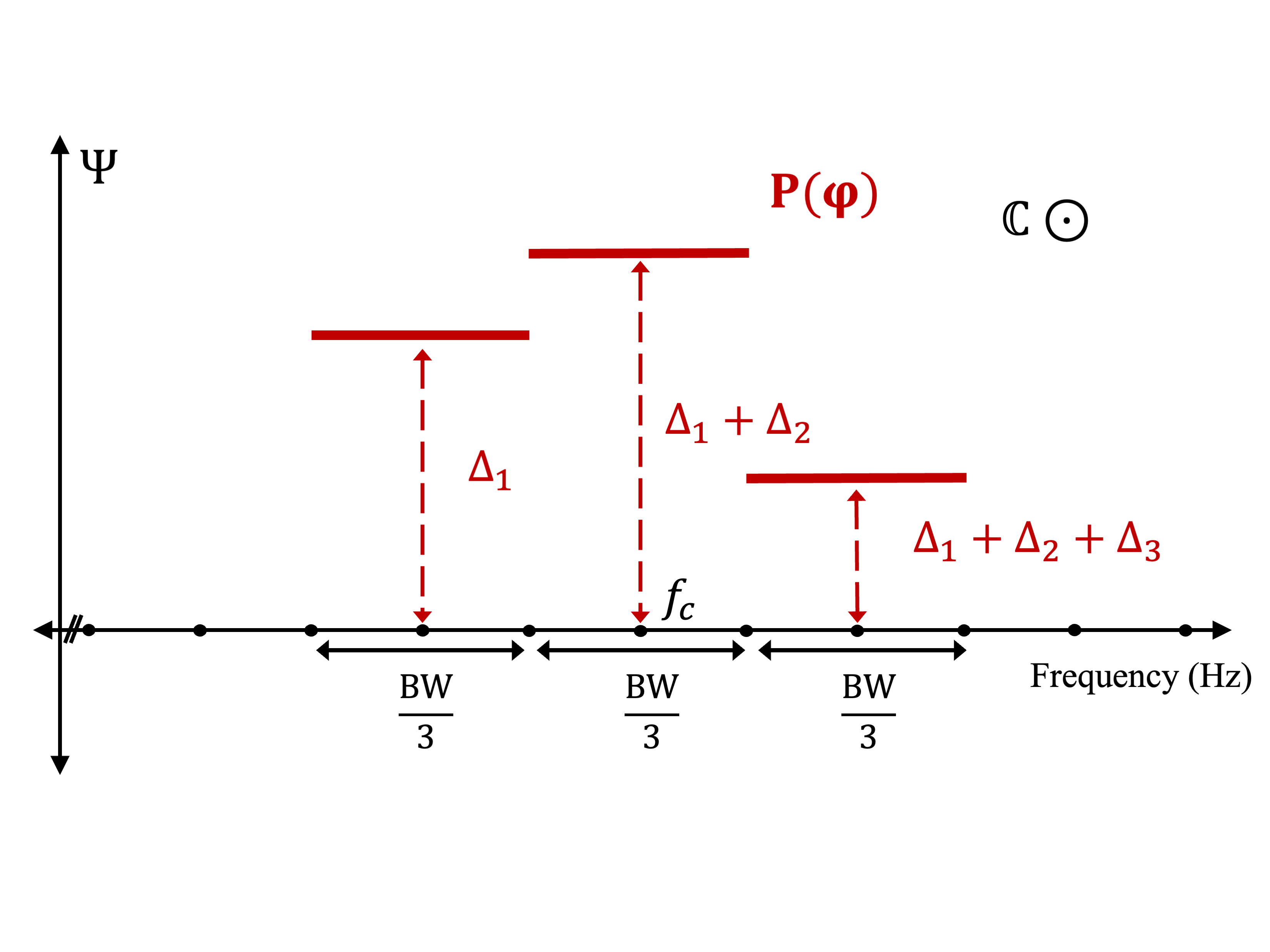}%
\label{Fig: 3 basis vectors result}}
\caption{Ideal split beampattern generators and beampattern generation process. (a) Ideal split beampattern generators created from shifted and scaled $2$ direction split beampatterns (b) Beampattern obtained by taking $\star$ operation of the generator beampatterns.}
\end{figure*}


\subsection{Split Beampattern Generators} \label{Section: Split Beampattern Generators}

We start our demonstration with the split beampatterns with $G=3$ partitions. Consider an arbitrary split beampattern $\mathbf{P}(\boldsymbol{\boldsymbol{\varphi}})$ with subband to direction mapping vector $\boldsymbol{\varphi}=[\varphi_1,\varphi_2,\varphi_3]^{H} \in \mathbb{R}^{3\times 1}$ and beampatterns $\mathbf{P}(\boldsymbol{\Delta_1}),\mathbf{P}(\boldsymbol{\Delta_2}),\mathbf{P}(\boldsymbol{\Delta_3})$ with the following direction mapping vectors: 
\begin{equation}
    \begin{aligned}
        &\boldsymbol{\Delta_1}=[\Delta_1,\Delta_1,\Delta_1]^{H} \in \mathbb{R}^{3\times 1}\\ 
        &\boldsymbol{\Delta_2}=[0,\Delta_2,\Delta_2]^{H}\in \mathbb{R}^{3\times 1}\\       
        &\boldsymbol{\Delta_3}=[0,0,\Delta_3]^{H} \in \mathbb{R}^{3\times 1}
    \end{aligned}
\end{equation}

Let $\mathbf{P}(\boldsymbol{\boldsymbol{\varphi}})=\mathbf{P}(\boldsymbol{\Delta_1})\star\mathbf{P}(\boldsymbol{\Delta_2})\star\mathbf{P}(\boldsymbol{\Delta_3})$ which is illustrated in Figure \ref{Fig: 3 basis vectors}, ignoring the out of bandwidth components which we discuss in section \eqref{Section: Approximating Generators}. Then, from Lemma \eqref{Lemma: addition of angles}:

\begin{equation}
    \mathbf{P}(\boldsymbol{\varphi})=\mathbf{P}(\boldsymbol{\Delta_1})\star\mathbf{P}(\boldsymbol{\Delta_2})\star\mathbf{P}(\boldsymbol{\Delta_3})=\mathbf{P}(\boldsymbol{\Delta_1+\Delta_2+\Delta_3}) \nonumber
\end{equation}

which implies that,

\begin{equation}
    [\varphi_1,\varphi_2,\varphi_3]^{H}=[\Delta_1,\Delta_1+\Delta_2,\Delta_1+\Delta_2+\Delta_3]^{H} 
    \label{Eqn: generator addition}
\end{equation}

Therefore, we can generate an arbitrary split beampattern with $\boldsymbol{\varphi}=[\varphi_1,\varphi_2,\varphi_3]^{H}\in \mathbb{R}^{3\times 1}$ from $3$ generators, $\mathbf{P}(\boldsymbol{\Delta_1}),\mathbf{P}(\boldsymbol{\Delta_2}),\mathbf{P}(\boldsymbol{\Delta_3})$, by solving the following simple linear equation for a given $\boldsymbol{\varphi}$:

\begin{equation}
\label{3_matrix_split_unwrap}
    \begin{bmatrix}
        1 & 0 & 0 \\
        1 & 1 & 0 \\
        1 & 1 & 1 \\
    \end{bmatrix} \begin{bmatrix}
        \Delta_1  \\
        \Delta_2  \\
        \Delta_3  \\
    \end{bmatrix}= \begin{bmatrix}
        \varphi_1  \\
        \varphi_2  \\
        \varphi_3  \\
    \end{bmatrix} 
\end{equation}

This reduces the beampattern design process to a simple linear equation with $3$ unknowns, which can be solved as follows:

\begin{equation}
\label{3_lin_split_unwrap}
\begin{aligned}
    &\Delta_1=\varphi_1\\ 
    &\Delta_2=\varphi_2-\Delta_1\\
    &\Delta_3=\varphi_3-\Delta_2
\end{aligned}
\end{equation}

From construction, generator beampatterns directions need to satisfy $\forall g>2,|\Delta_{g}|<1$ to avoid wrapping. Therefore, if $\exists g>2, |\Delta_{g}|>1$ we revise the generator beampatterns as follows:

\begin{equation}
    \begin{aligned}
    &\boldsymbol{\Delta_g} \leftarrow \boldsymbol{\Delta_g}- \lfloor\frac{\Delta_g}{2}\rfloor \\   
    &\boldsymbol{\Delta_1}  \leftarrow \boldsymbol{\Delta_1}+ \lfloor\frac{\Delta_g}{2}\rfloor 
    \end{aligned}
    \label{Eqn: unwrap_angles}
\end{equation}

which does not change the resulting split beampattern from the equation \eqref{Eqn: generator addition}. The above process can be generalized to $G>3$ partitions by defining the generator beampatterns, $\mathbf{P}(\boldsymbol{\Delta_g}),\;g\in\;[1,G]$ with the following subband-direction mapping vector:

\begin{equation}
    \boldsymbol{\Delta_g}=[\underbrace{0,\dots,0}_{g-1},\underbrace{\Delta_g,\dots,\Delta_g}_{G-g+1}]\;g\in\;[1,G]
\end{equation}

and solving the following linear equation:

\begin{equation}
\label{matrix_split_unwrap}
    \begin{bmatrix}

         1 & 0 & 0 \\
         1 & \ddots & 0 \\
         1 & 1 & 1 \\
    \end{bmatrix} \begin{bmatrix}
        \Delta_1  \\
        \vdots  \\
        \Delta_G  \\
    \end{bmatrix}= \begin{bmatrix}
        \varphi_1  \\
        \vdots \\
        \varphi_G  \\
    \end{bmatrix} 
\end{equation}

which the solution can be found by the following recursion:

\begin{equation}
\begin{aligned}
    &\Delta_1=\varphi_1\\
    &\Delta_g=\varphi_g-\Delta_{g-1};\; g \in [2,G]
\label{Eqn: HDB_recursion}
\end{aligned}
\end{equation}

The corresponding generator beampatterns are revised as in equation \eqref{Eqn: unwrap_angles}.

We have shown that a split beampattern $\mathbf{P}(\varphi)$ can be represented over $\star$ operation with split beampattern generators $\mathbf{P}(\boldsymbol{\Delta_g}),\;g\in\;[1,G]$ as follows:

\begin{equation}
\begin{aligned}
    \mathbf{P(\varphi)}=\mathbf{P}(\boldsymbol{\Delta_g}) \star \dots \star \mathbf{P}(\boldsymbol{\Delta_G}) \label{eqn: split star representation}
\end{aligned}
\end{equation}

where $\Delta_g,\;g\in\;[1,G]$ are obtained from equations \eqref{Eqn: HDB_recursion},\eqref{Eqn: unwrap_angles}. Split beampattern generators for $g>1$, like the desired split beampattern, cannot be realized by \gls{TTD} arrays due to the discontinuity of directions between different subbands \cite{jpta}. However, we have successfully represented the desired split beampattern with split beampattern generators; if split beampattern generators are easy-to-approximate, we can utilize the divide-and-conquer strategy discussed in section \eqref{Section: Divide and conquer}. We continue with the properties of the split beampattern generators and show that such generators can be efficiently approximated with \gls{TTD} arrays.

\begin{figure*}[!t]
\centering
\subfloat[$\mathbf{P}_{\boldsymbol{\Phi_{\boldsymbol{\Delta_{2}}}},\; \boldsymbol{\Delta_{2}}}\approx{[0,0.2,0.2]^H}$]{\includegraphics[width=0.33\textwidth]{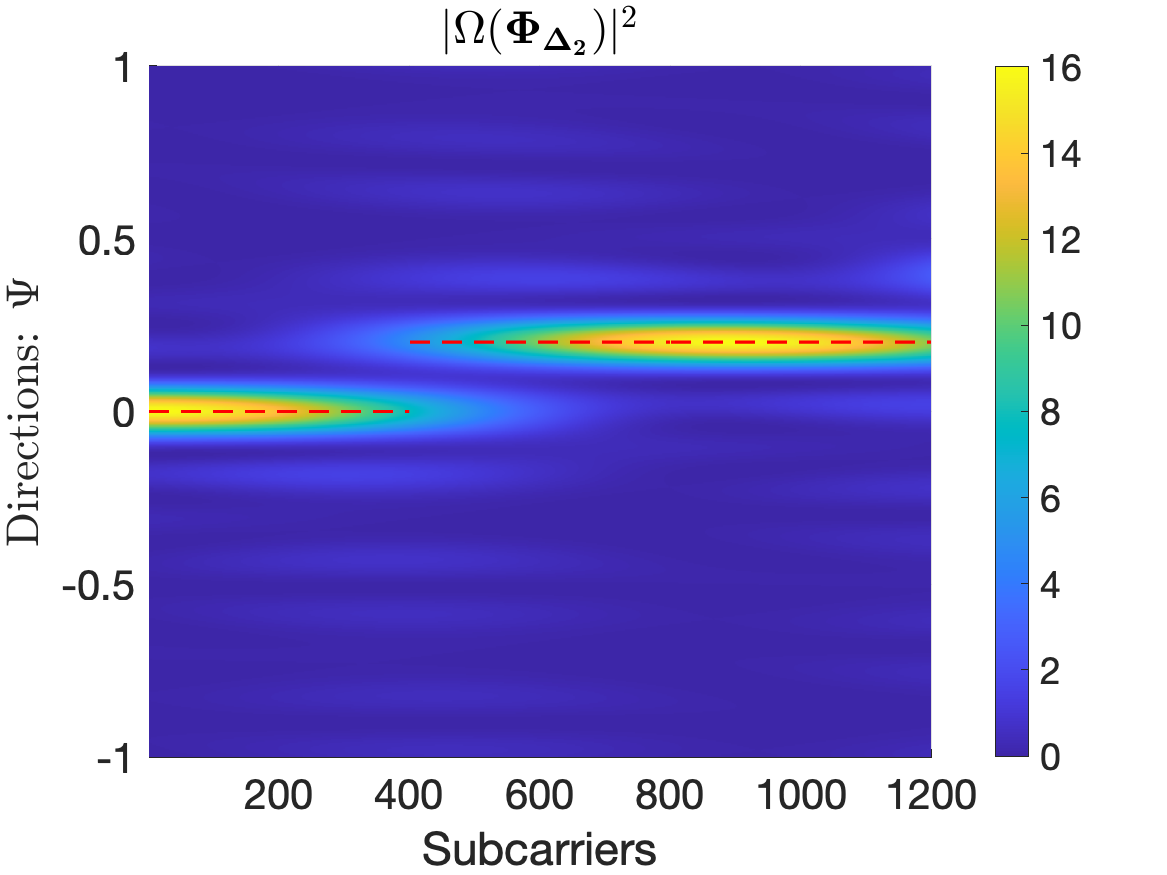}%
\label{HDB_generator_1}}
\hfil
\subfloat[$\mathbf{P}_{\boldsymbol{\Phi_{\boldsymbol{\Delta_{3}}}},\; \boldsymbol{\Delta_{3}}}\approx{[0,0,0.2]^H}$]{\includegraphics[width=0.33\textwidth]{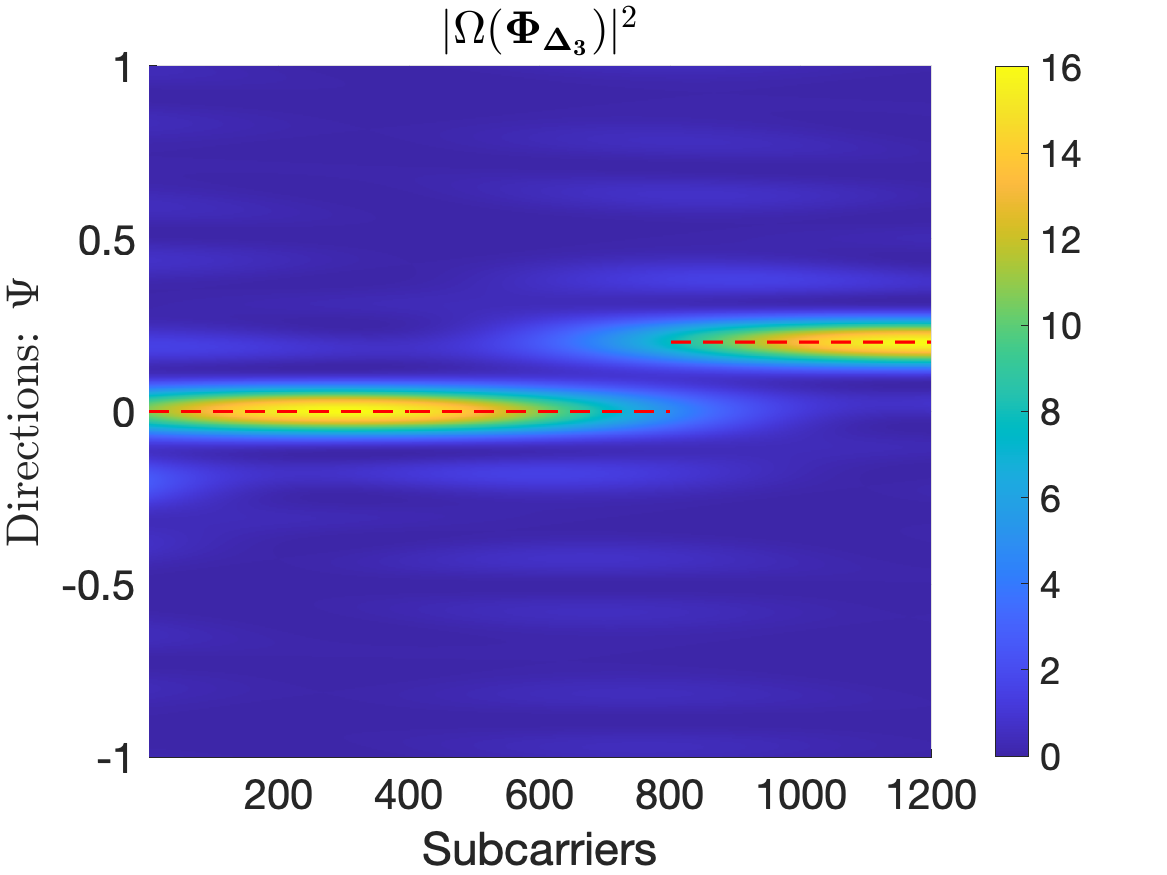}%
\label{HDB_generator_2}}
\hfil
\subfloat[$\mathbf{P}_{\boldsymbol{\Phi_{\boldsymbol{\Delta_{2}}+\boldsymbol{\Delta_{3}}}},\;\boldsymbol{\Delta_{2}}+\boldsymbol{\Delta_{3}}}\approx{[0,0.2,0.4]^H}$]{\includegraphics[width=0.33\textwidth]{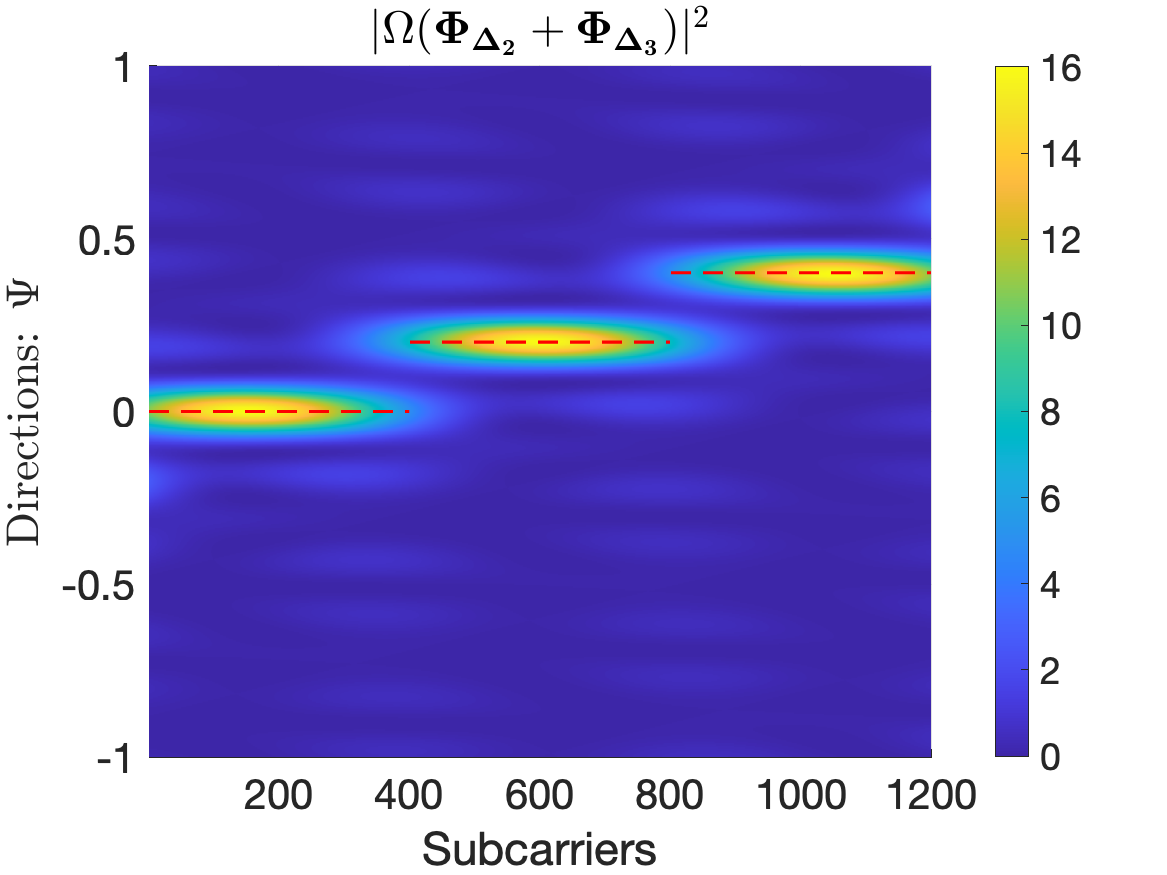}
\label{HDB_aprox}}
\caption{Generator Beampatterns approximated by the mmFlexible math \cite{mmflexible} algorithm (\ref{HDB_generator_1},\ref{HDB_generator_2}) and resulting beampattern after adding corresponding time delay and phase shifts (\ref{HDB_aprox}). Dashed red lines indicate the target directions. }
\end{figure*}

\subsection{Approximating Split Beampattern Generators} \label{Section: Approximating Generators}

Split beampattern generators $\mathbf{P}(\boldsymbol{\Delta_g}),\;g\in\;[1,G]$ can be directly approximated with \gls{TTD} arrays $\mathbf{P}(\boldsymbol{\Delta_g}) \approx \mathbf{P}_{\Phi_{\Delta_g}}(\boldsymbol{\Delta_g}) \in \mathcal{P}_{\mathcal{T}}$ with the optimization based algorithms \cite{jpta,mmflexible} discussed in Section \eqref{Section: Problem Definition}. However, optimization problem-based approaches do not provide any computational advantage as approximating generators have the same computational complexity as approximating the desired split beampattern. Fortunately, we can utilize the structure of the generator beampatterns to simplify the approximation process significantly.

Similar to Section \eqref{Section: Split Beampattern Generators}, we first demonstrate the structure of the split beampattern generators for $G=3$ subbands; which are illustrated in Figure \eqref{Fig: 3 basis vectors}. We can observe that $\mathbf{P}(\boldsymbol{\Delta_1})=\mathbf{P_{\Phi_{\Delta_1}}}(\boldsymbol{\Delta_1})$ has the same direction assignment for all subbands. Therefore, it can be realized by \gls{TTD} arrays, in closed form, with the following array configuration matrix \cite{TTT_rainbow_hardware_WSU_II}:

\begin{equation}
\begin{aligned}
      & \boldsymbol{\Phi_{\Delta_1}}=\begin{bmatrix}
        \frac{-\Delta_1 \mathbf{n}}{2 f_c} & \boldsymbol{0}
    \end{bmatrix} \in \mathcal{V}_1\\ 
    &\mathbf{n}=\begin{bmatrix}
        0 & \dots & N-1
\end{bmatrix}^H  
\end{aligned} \label{Eqn: Constant beampattern}
\end{equation}

Furthermore, $\mathbf{P}(\boldsymbol{\Delta_2})$ can be viewed as split beampattern with $G=2$ subbands for center frequency $f^{\Delta_1}_c=f_c-\frac{BW}{2}+\frac{BW}{3}$ and with enlarged bandwidth of $BW^{\Delta_1}=\frac{4BW}{3}$. Therefore, it can be approximated with low complexity $G=2$ subband split beampattern approximation algorithms \cite{structured,modulo,mmflexible,jpta} for appropriate center frequency and bandwidth. This approach does not affect the desired split beampattern as subcarriers beyond the assigned bandwidth can be ignored. Similarly, split beam generators for $G>3$ subbands can be approximated as $G=2$ split beampatterns with the following center frequency and bandwidth:

\begin{equation}
    \begin{aligned}
    & f_{c}^{\Delta_g}=f_c-\frac{BW}{2}+\frac{(g-1)BW}{G}\\
    & BW_{\Delta_g}=2 BW \frac{G-1}{G}\\
    \label{Eqn: fc and BW of generators}
    \end{aligned}
\end{equation}

The split beampattern generator approximation can further be simplified by utilizing the following observation: changing center frequency and bandwidth of a $G=2$ split beampattern to $f^{\Delta_g}$ and $BW^{\Delta_g}$ can be accomplished by solving the following linear equation \cite{modulo}:

\begin{equation}
    \begin{aligned}
        & \forall m:\;-2\pi f_m\mathbf{t}-\boldsymbol{\phi} = -2\pi f^{\Delta_g}_m \mathbf{t}^{\Delta_g}-\boldsymbol{\phi}^{\Delta_g}\\
        &\forall m:\;-2\pi (f_c+\frac{mBW}{M}-\frac{BW}{2})\mathbf{t}-\boldsymbol{\phi} \\
        & \qquad \quad = -2\pi(f^{\Delta_g}_c+\frac{mBW^{\Delta_g}}{M}-\frac{BW^{\Delta_g}}{2})  \mathbf{t}^{\Delta_g}-\boldsymbol{\phi}^{\Delta_g}
    \end{aligned}
\end{equation}

which is obtained from precoder matrix definition in equation \eqref{Def: Omega1} and can be solved as:

\begin{equation}
    \begin{aligned} 
    & \mathbf{t}^{\Delta_g}=\frac{1}{\alpha_g} \mathbf{t}\\
    &\boldsymbol{\phi}^{\Delta_g}=\boldsymbol{\phi}-2 \pi f_c \mathbf{t}+\frac{1}{\alpha_g}2 \pi f_{c}^{\Delta_g} \mathbf{t}\\
    & \alpha_g=\frac{BW^{\Delta_g}}{BW}
    \end{aligned}
    \label{Def: scale and shift}
\end{equation}

Therefore, we can approximate split beampattern generators $\mathbf{P}(\boldsymbol{\Delta_g}),\forall g$ from the approximation of split beampatterns with $G=2$ subbands with a \textbf{linear transformation} over corresponding array configuration matrices, as in equation \eqref{Def: scale and shift}. Hence, If we define a dictionary of array configuration matrices of split beampatterns with $G=2$ subbands, we can approximate any split beampattern generator with a simple linear transformation. We continue with the construction of the $G=2$ subband split beampattern dictionary, which we call \textit{generator dictionary}.

\subsection{Generator Dictionary} \label{Section:Generator Dictionaries}

We construct the generator dictionary, i.e., dictionary of $G=2$ partition split beampatterns $\boldsymbol{\Upsilon} \in \mathbb{R}^{N \times 2 \times 2A}$ for $A$ grid directions in two steps. First, we obtain the array configuration matrices $\boldsymbol{\Phi}_{\underline{\Delta}}$ of approximate split beampatterns with subband to direction mapping $\underline{\boldsymbol{\Delta}}=[\Delta_1,\Delta_2],\; \Delta \in [-1,1]$ with JPTA approximation algorithm \cite{jpta} or with equation \eqref{Eqn: Constant beampattern} if $\Delta_1=\Delta_2$. The fidelity of the approximated split beampatterns is directly related to the fidelity of the approximated split beampatterns. As illustrated in figures \eqref{HDB_generator_1}, gain across subcarriers of each subband is not constant, and the maximum is not necessarily attained at the center of the subband. Hence as a second step, we post-process the split beampatterns such that the maximum gain is attained at the center of each subband by adjusting the bandwidth as $BW^{\underline{\Delta}}=BW\frac{M}{2|\arg\max |[\mathbf{P}_{\boldsymbol{\Phi}_{\underline{\Delta}}}]_{(\Delta_2,:)}|-\arg\max |[\mathbf{P}_{\boldsymbol{\Phi}_{\underline{\Delta}}}]_{(\Delta_1,:)}||}$ with the equation \eqref{Def: scale and shift}. We empirically observe that this post-processing results in uniform gain across subcarriers of the approximated split beampatterns.

Then, utilizing the dictionary $\boldsymbol{\Upsilon}$, we approximate any generator beampattern $\mathbf{P}(\boldsymbol{\Delta_g})$ by simply reading the dictionary for the corresponding directions and scaling its center frequency and bandwidth with equations \eqref{Eqn: fc and BW of generators} and \eqref{Def: scale and shift}.

\subsection{Homomorphic Directional Beamforming Algorithm}

We have shown that split beam generators can be efficiently approximated simply by reading the corresponding \gls{TTD} array configuration from the dictionary $\boldsymbol{\Upsilon}$ followed by a trivial linear transformation defined in equation \eqref{Def: scale and shift}. Since the desired split beampattern is represented by generator beampatterns as in equation \eqref{eqn: split star representation} and approximated generator beampatterns are in \gls{TTD} beampatterns group, we can utilize the observed homomorphism and obtain array configuration matrix of the desired split beampattern from the equation \eqref{eqn: config matrix represnt.}. Therefore, instead of approximating the desired split beampattern directly, we efficiently approximated its generators and obtained the corresponding array configuration matrix from the observed homomorphism.

This basically reduces the split beampattern generation process to solving simple linear equations (\ref{matrix_split_unwrap}) and adding the time delay and phase shifter values of the corresponding generators as in equation \eqref{eqn: time and phase representation}. \gls{HDB} algorithm is summarized in algorithm \eqref{Alg: HDB}. Figures (\ref{HDB_generator_1},\ref{HDB_generator_2}) illustrate the approximated generators $\mathbf{P}(\boldsymbol{\Delta_{2}})$ and $\mathbf{P}(\boldsymbol{\Delta_{3}})$ for $G=3$ with mmFlexible algorithm \cite{mmflexible} and figure (\ref{HDB_aprox}) illustrates the resulting beampattern $\mathbf{P}(\boldsymbol{\Delta_2}+\boldsymbol{\Delta_3})$ after adding the array configuration matrices of the approximated generators.

\begin{algorithm}\label{Alg: HDB}

	\KwIn{Desired split beampattern $\mathbf{P}(\boldsymbol{\varphi})$ with $\boldsymbol{\varphi}=[\varphi_1,\dots,\varphi_G]^{H} \in \mathbb{R}^{G \times 1}$} 
	\KwOut{Array configuration vector $\boldsymbol{\Phi}_{\boldsymbol{\varphi}}$ such that $\Omega_3(\boldsymbol{\Phi}_{\boldsymbol{\varphi}})=\mathbf{P}_{\boldsymbol{\Phi}_{\boldsymbol{\varphi}}} \approx \mathbf{P}(\boldsymbol{\varphi})$ }

	Solve the recursion in the equation (\ref{Eqn: HDB_recursion}) and obtain $\Delta_g, \;g\in [1,G]$ from the equation \eqref{Eqn: unwrap_angles}
 
    Obtain $\boldsymbol{\Phi_{\Delta_{g}}},\;\forall g$ that approximates the generator $\mathbf{P}_{\boldsymbol{\Delta_{g}}}$ from the dictionary $\boldsymbol{\Upsilon}$ with equation \eqref{Eqn: fc and BW of generators} and \eqref{Def: scale and shift}

    Add time delay and phase shifter vectors, equation (\ref{eqn: time and phase representation}): $\mathbf{t}=\sum_{l=1}^{G} \mathbf{t}^{g_{l}},\boldsymbol{\phi}=\sum_{l=1}^{G} \boldsymbol{\phi}^{g_{l}}$

    \Return{$\boldsymbol{\Phi}_{\boldsymbol{\varphi}}=\begin{bmatrix}
        \mathbf{t} & \boldsymbol{\phi}
    \end{bmatrix},$}
    
	\caption{\gls{HDB} Algorithm}
\end{algorithm}

\subsection{Effect of Generator Approximation}

As discussed, the performance of the \gls{HDB} algorithm depends on the fidelity of the approximated generators. As the \gls{HDB} algorithm effectively relies on the addition of the subband directions of each split beampattern generator, any direction error inherited from dictionary construction, such as errors from the JPTA approximation algorithm \cite{jpta} due to the increasing number of antennas or bandwidth, can accumulate. We investigate the effect of error accumulation in Section \eqref{Section: Performance Analysis} through extensive simulations. Furthermore, the approximation of generator beampatterns is empirically observed \cite{jpta} and systematically analyzed \cite{structured} to result in the weighted addition of two directional beams as shown in figures \eqref{HDB_generator_1}, notably observable at the transition between subbands. We design the generator beampatterns to avoid coinciding transition regions of different generator beampatterns so that approximated generators preserve the split beampattern generator's behavior.

\section{Performance Analysis} \label{Section: Performance Analysis}

\begin{table}[t!]
\caption{Simulation Parameters}
\label{table:simulation_parameters}
\setlength{\tabcolsep}{3pt}
\begin{tabular}{|p{155pt}|p{75pt}|} 
\hline
\textbf{\textsc{Parameter}}& \textbf{\textsc{Value}} \\
\hline
Center frequency $f_c$ & $28$ GHz \\
Bandwidth $BW$ & $1:10$ GHz \\
Pre-beamforming SNR & $10$ dB\\
Number of subcarriers $M$ & $1200$ \\
Number of \gls{ULA} antenna $N$ & $[8,16,32,64,128]$  \\
Number of \gls{UE}s $G$ & $[3,4,5,6,7,8]$\\
\gls{UE} direction grid numbers $A$ & $499$\\
\gls{UE} directions $\varphi_d,\;d \in [1,G]$ & $\sim \mathcal{U}_{A}[-1,1]$ \\
Number of Monte Carlo Locations & $5000$ \\
JPTA \cite{jpta} Approximation Maximum Delay & $\frac{M}{BW}=3.6 \; \mu s$\\
JPTA \cite{jpta} Approximation iteration $I_A$ & $30$\\
\hline
\end{tabular}
\end{table}

This section demonstrates a detailed analysis of the performance of \gls{HDB} algorithm compared to the benchmark algorithms. Numerical simulations are conducted for an \gls{OFDM} system with an \gls{ULA} comprising $N=16$ antenna, operating at a carrier frequency of $f_c=28$ GHz with bandwidth $BW=3$ GHz. The system has $M=1200$ subcarriers and serves $G=3$ \gls{UE}s with pre-beamforming SNR of $10$ dB unless otherwise stated. Monto Carlo simulations are conducted over $5000$ different \gls{UE} direction configurations, where \gls{UE} directions are drawn from discrete uniform random variable over the domain $[-1,1]$ with $A=499$ discrete directions i.e. $\varphi_d \sim \mathcal{U}_{A}[-1,1]$ for every \gls{UE} $d \in [1,G]$. The system configuration and benchmark algorithm parameters are summarized in Table \eqref{table:simulation_parameters}, and the following benchmark schemes are discussed for comparison:
\\

\begin{itemize}
    \item \textbf{mmFlexible \cite{mmflexible} Math}: Convex approximation based solution of the problem \eqref{eq: obj_fun_2}, ignores beam-squint. The algorithm provides a closed-form solution with time complexity $\mathcal{O}(NG)$ where $G$ is the number of \gls{UE}s. 
    \item \textbf{FSDA \cite{mmflexible}}: Exhaustive search-based solution of the problem \eqref{eq: obj_fun_2}, ignores beam-squint. 
    \item \textbf{JPTA \cite{jpta} Line Search}: Exhaustive search-based solution of the problem \eqref{eq: obj_fun_1}.
    \item \textbf{JPTA \cite{jpta} Approximation}: Convex approximation based solution of the problem \eqref{eq: obj_fun_1}. The algorithm iteratively solves the approximated problem with a $\mathcal{O}(MN)$ complexity in each iteration.
\end{itemize}
    
The \textit{Spectral Efficiency (SE)} is considered as a performance metric which is defined as follows for the target split beampattern with subband direction mapping vector $\boldsymbol{\varphi}=\mathbb{C}^{G \times 1}$:

\begin{equation}
    \text{SE}(m,\boldsymbol{\varphi})=\log_2(1+\frac{|[\overline{\mathbf{P_{\varphi}}}]_{([\boldsymbol{\varphi}]_{\lceil\frac{m}{M'}\rceil},m)}|^2}{\sigma^2})
    \label{eqn: spectral efficiency}
\end{equation}

where $\overline{\mathbf{P_{\varphi}}}$ is the synthesized beampattern and $\sigma^2$ is the variance of the noise, assuming no-pathloss $\sigma^2=\frac{1}{SNR}$.


\begin{table*}[t!]
\caption{Practical Deployment Requirements of Proposed the \gls{HDB} and Benchmark Algorithms.}
 \label{tab:compexity}
    \centering
\begin{tabular}{||c|| c| c| c|c||} 
 \hline
Algorithms & Computational Complexity  & Dict. Size & Dict. Memory for 3 \gls{UE} & Average Runtime for 3 \gls{UE} \\ 
 \hline\hline
Proposed \gls{HDB} & $\mathcal{O}(NG)$ &  $\mathbb{R}^{2\times N \times 2A}$ & $\approx 62$ KB & $2.4 \times 10^{-8}$ s\\ 
 \hline
JPTA Approximation \textbf{[Online]} \cite{jpta}  & $\mathcal{O}(MN I_A)$ & 0 & $0$ & $4.7\times 10^{-2}$ s\\
 \hline
 mmFlexible Math \textbf{[Online]} \cite{mmflexible}  & $\mathcal{O}(NG)$ & $0$ & 0  & $3.7 \times 10^{-8}$ s \\ 
 \hline
 JPTA Approximation \textbf{[Dictionary]} \cite{jpta}  & $\mathcal{O}(1)$ & $\mathbb{R}^{2\times N \times A^G}$ & $ \approx 7$ GB & NR \\
 \hline
  JPTA Line Search \textbf{[Dictionary]} \cite{jpta}   & $\mathcal{O}(1)$& $\mathbb{R}^{2\times N \times A^G}$ &$ \approx 7$ GB & NR \\
 \hline
FSDA \textbf{[Dictionary]} \cite{mmflexible}  & $\mathcal{O}(1)$ & $\mathbb{R}^{2\times N \times A^G}$ & $ \approx 7$ GB  & NR \\
 \hline

\end{tabular}
\end{table*}

\subsection{Computational Complexity and Memory Analysis}

In this section, we investigate the computational complexity and memory requirements of \gls{HDB} and benchmark algorithms. The proposed \gls{HDB} algorithm has the complexity of $\mathcal{O}(NG)$ and requires a dictionary of size $\mathbb{R}^{2\times N \times 2A}$. We discuss two different deployment scenarios of the benchmark algorithms: \textit{online} and \textit{dictionary-based}. The \textit{online} deployment refers to the case where algorithms calculate array configuration matrices in real-time, and \textit{dictionary-based} deployment refers to the case where the array configuration matrices are obtained from the pre-calculated configuration matrix dictionary of all possible \gls{UE} configurations. FSDA \cite{mmflexible} and JPTA line search \cite{jpta} algorithms are exhaustive search algorithms and are impractical for \textit{online} deployment; therefore, we only discuss \textit{dictionary-based} implementation of these algorithms in the following sections and table \eqref{tab:compexity}. 

The \textit{dictionary-based} deployment of the JPTA \cite{jpta} and mmFlexible \cite{mmflexible} algorithms require the storage of the array configuration matrix of all possible \gls{UE} direction configurations, amounting to $A^G$ possible cases. This corresponds to a matrix of size $\mathbb{R}^{2 \times N \times A^G}$. Even for a small number of \gls{UE}s, such as for $G=3$ and for system configuration at table \eqref{table:simulation_parameters}, the required dictionary size is $\approx 7$ GB for $64$ bit floating-point variables, and it increases \textbf{exponentially} with the number of \gls{UE}s, making \textit{dictionary-based} deployment impractical. On the contrary, the \gls{HDB} algorithm only requires a negligible $62$ KB for the same system configuration, and the size of the required memory is \textbf{independent} of the number of \gls{UE}s. The dictionary size and memory requirements are summarized in the table \eqref{tab:compexity}.

As the \textit{dictionary-based} deployment is shown to be impractical, we continue with \textit{online} deployment cases of the benchmark algorithms. In such a scenario, the computational complexity of the JPTA Approximation algorithm \cite{jpta} is $\mathcal{O}(N M I_A)$ and the computational complexity of both \gls{HDB} and mmFlexible math algorithms is $\mathcal{O}(NG)$. Since the number of subcarriers is much larger than the number of \gls{UE}s, i.e., $M>>G$, the computational complexity of the JPTA Approximation \cite{jpta} is expected to be prohibitively higher.

To further assess the computational requirements, we analyze the average runtime of the \gls{HDB} and \textit{online} deployment cases of the benchmark algorithms. For simulations, we utilize MATLAB \textit{version 23.2.0.2428915 (R2023b) Update 4} on the \textit{Linux Ubuntu} operating system comprising \textit{AMD Ryzen Threadripper 3970X 32-Core} Processor. Figure \eqref{Fig: runtime} shows the average runtime for different numbers of \gls{UE}s, and the table \eqref{tab:compexity} demonstrates the average runtime for $G=3$ \gls{UE}s. Consistent with the complexity analysis above, the JPTA Approximation \cite{jpta} algorithm, even for a single iteration, is orders of magnitude more computationally demanding than mmFlexible math \cite{mmflexible} and the proposed \gls{HDB} algorithms. Both \gls{HDB} and mmFlexible math \cite{mmflexible} algorithms perform similarly with very low runtime, making both algorithms low-complexity and practical alternatives to the JPTA Approximation algorithm \cite{jpta}. 

\begin{figure}[t!]
\centering
\includegraphics[width=0.50\textwidth]{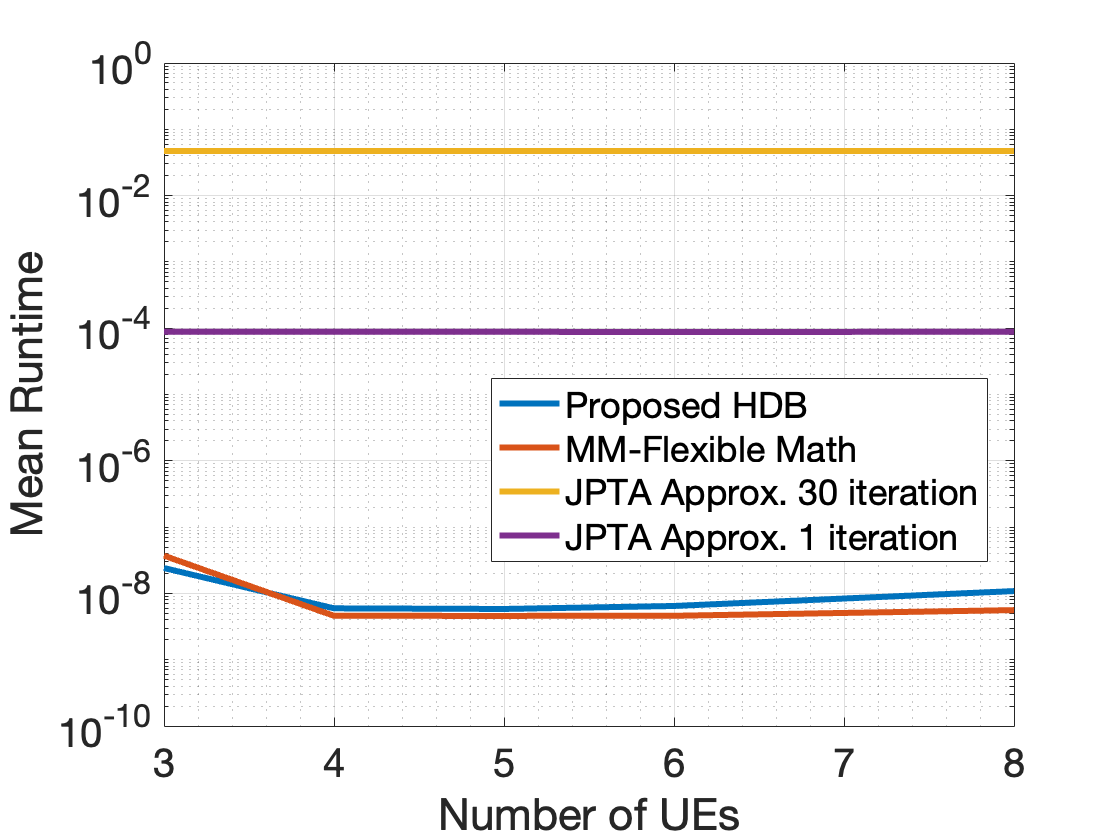}
\caption{Average runtime of the proposed \gls{HDB} algorithm and benchmark algorithms for different numbers of \gls{UE}s, $G$, for $N=16$ antennas and bandwidth $BW=3$ GHz. Exhaustive search algorithms are not included, and single iteration, $I_A=1$, JPTA Approximation algorithm \cite{jpta} is included for comparison.}
\label{Fig: runtime}
\end{figure}

\subsection{Performance Analysis Over Subcarriers}

We continue our analysis with the evaluation of the \textit{spectral efficiency} given in the equation \eqref{eqn: spectral efficiency} of the \gls{HDB} and \textit{online} deployment of the benchmark algorithms for various parameters, including the number of \gls{UE}s, antennas, and bandwidth.

Figure \eqref{Fig: bar_3_users} shows the \gls{ASE} over different subbands for $3$ \gls{UE}s of $5000$ Monte Carlo realizations. We observe that all algorithms perform very close to the maximum \gls{ASE} at subbands $1$ and $3$ where JPTA approximation \cite{jpta}, mmFlexible math \cite{mmflexible}, and \gls{HDB} algorithms achieve approximately $95\%,93\%,90\%$ of the maximum \gls{ASE}, respectively. However, in the subband $2$, JPTA approximation algorithm \cite{jpta} with $I_A=30$ iterations and mmFlexible math algorithm \cite{mmflexible} achieves $83.35\%$ and $84.53\%$ of the maximum \gls{ASE}, providing unfair \gls{ASE} compared to the users assigned to subbands $1$ and $3$. Additionally, we observe that the JPTA approximation algorithm \cite{jpta} has improved performance at subband $2$ with the increasing number of iterations. On the other hand, \gls{HDB} algorithm achieves $89.55\%$ of the maximum \gls{ASE} at the subband 2, providing similar and near-optimal performance across different subbands.

\begin{figure}[t!]
\centering
\includegraphics[width=0.50\textwidth]{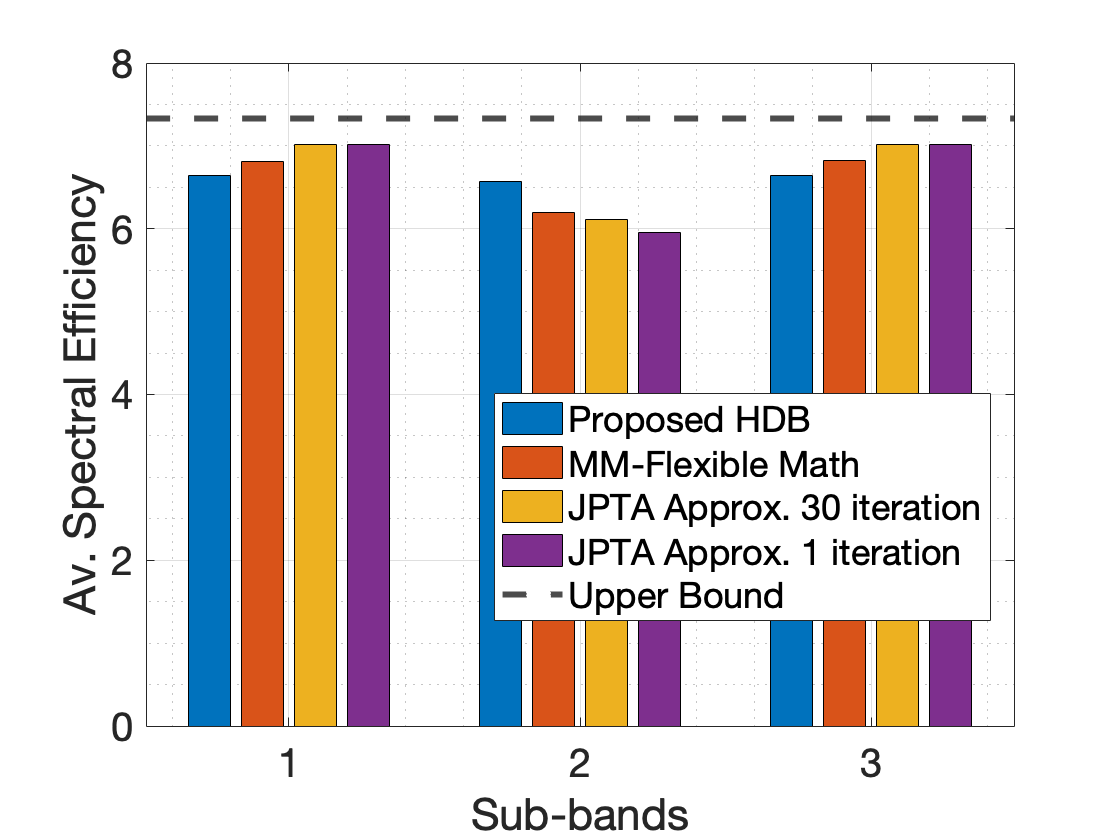}
\caption{\gls{ASE} of the proposed \gls{HDB} algorithm and benchmark algorithms across subbands for $G=3$ \gls{UE}s, $N=16$ antennas and bandwidth of $BW=3$ GHz. The \textit{Upper Bound} refers to the \gls{ASE} achieved with the maximum beamforming gain.}
\label{Fig: bar_3_users}
\end{figure}

\begin{figure}[t!]
\centering
\includegraphics[width=0.50\textwidth]{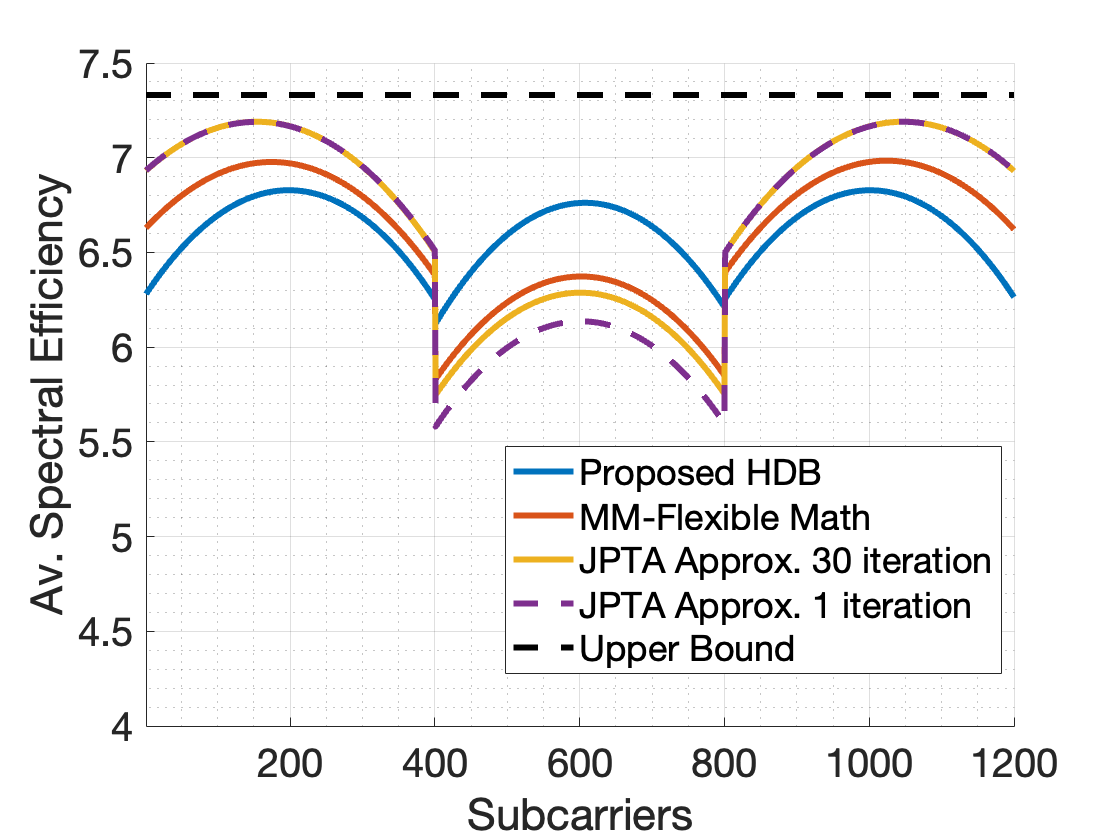}
\caption{\gls{ASE} of the proposed \gls{HDB} algorithm and benchmark algorithms across different subcarriers for $G=3$ \gls{UE}s, $N=16$ antennas and bandwidth of $BW=3$ GHz. The \textit{Upper Bound} refers to the \gls{ASE} achieved with the maximum beamforming gain.}
\label{Fig: gain_3_users}
\end{figure}

\begin{figure}[t!]
\centering
\includegraphics[width=0.50\textwidth]{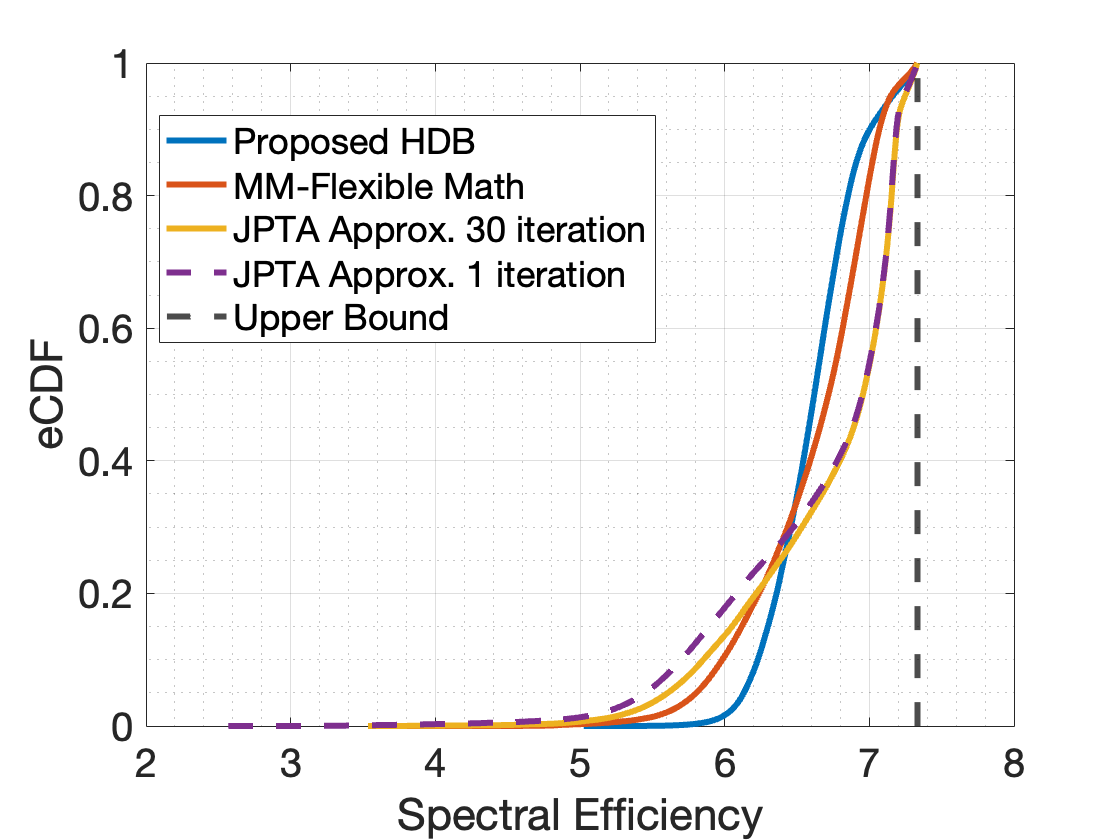}
\caption{Empirical cumulative distribution of the spectral efficiency across different subcarriers and \gls{UE} direction configurations for the proposed \gls{HDB} algorithm and benchmark algorithms for $3$ \gls{UE}s for $G=3$ \gls{UE}s, $N=16$ antennas and bandwidth of $BW=3$ GHz. The \textit{Upper Bound} refers to the \gls{ASE} achieved with the maximum beamforming gain.}
\label{Fig: ecdf_3_users}
\end{figure}

Figure \eqref{Fig: gain_3_users} demonstrates the \gls{ASE} over different sub-carriers for $3$ \gls{UE}s. We observe that subcarriers close to the center of each subband have higher spectral efficiency compared to subcarriers at the edge of each subband. JPTA approximation algorithm \cite{jpta} has a slightly off-center maximum at each subband and has a significant difference between \gls{ASE} across subcarriers: the maximum SE is $25\%$ and $28\%$ higher than the minimum SE for $I_A=30$ and $I_A=1$ iterations, respectively. Similarly, the mmFlexible math \cite{mmflexible} algorithm achieves the maximum at the slightly off-center of each subband, and the maximum SE is $19\%$ higher than the minimum SE. On the contrary, the proposed \gls{HDB} algorithm achieves its maximum close to the center for all subbands, with the maximum SE being only $7\%$ higher than the minimum SE. Results indicate that the proposed \gls{HDB} algorithm provides a more uniform \gls{ASE} distribution across subcarriers.

Figure \eqref{Fig: ecdf_3_users} shows the empirical cdf of the \gls{ASE} across different sub-carriers and \gls{UE} locations for $3$ \gls{UE}s. Across all realizations, i.e., subcarriers and \gls{UE} locations, \gls{HDB} algorithm provides higher minimum \gls{ASE} and less variance compared to benchmark algorithms. JPTA approximation $I_A=30$ iteration and mmFlexible math algorithms achieve less than $6$ bps/Hz spectral efficiency for $13\%$ and $10\%$ of the realizations, while \gls{HDB} achieves less than $6$ bps/Hz spectral efficiency for $1\%$ of the realizations. Results demonstrate that \gls{HDB} algorithm provides consistent spectral efficiency across all realizations, making it a reliable alternative to the benchmark algorithms.

\subsection{Impact of Array Parameters}

In this section, we assess the \gls{HDB} and benchmark algorithms for different array parameters and number of \gls{UE}s.

\begin{figure}[t!]
\centering
\includegraphics[width=0.50\textwidth]{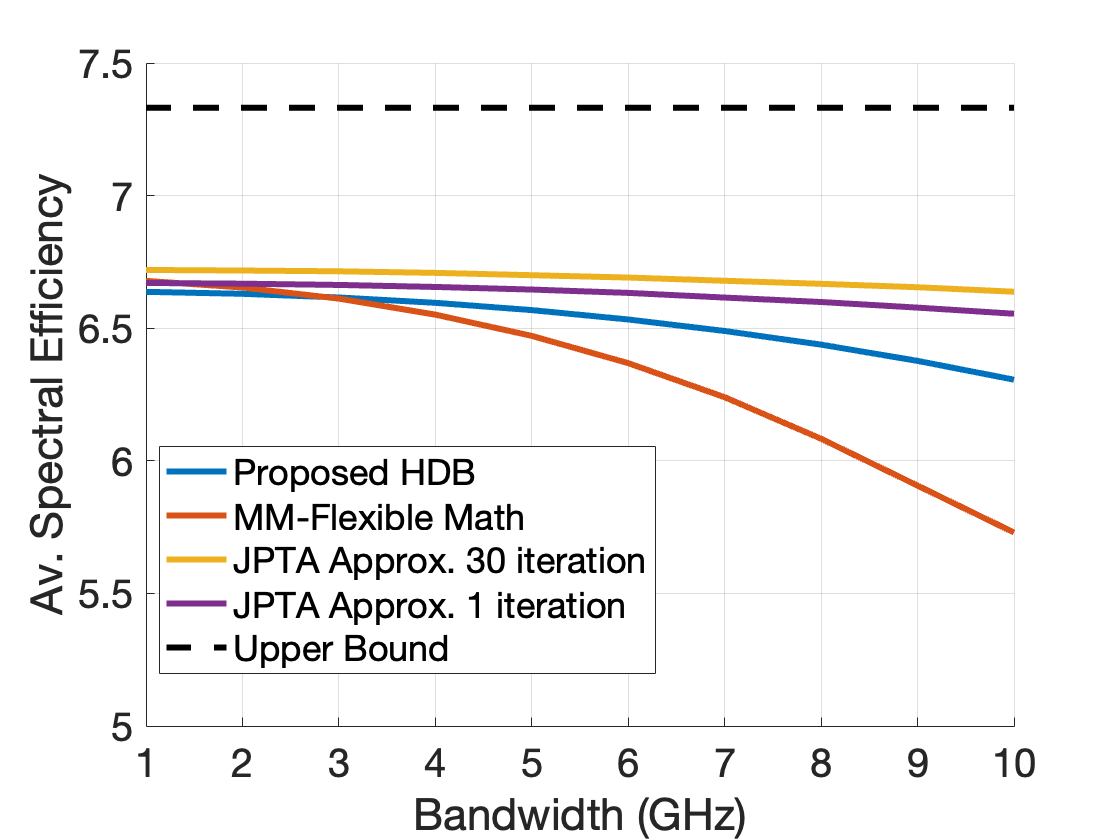}
\caption{\gls{ASE} of the proposed \gls{HDB} algorithm and benchmark algorithms for different bandwidths $BW$ for $G=3$ \gls{UE}s, $N=16$ antennas. The \textit{Upper Bound} refers to the \gls{ASE} achieved with the maximum beamforming gain.}
\label{Fig: BW 3 users}
\end{figure} 

\begin{figure}[t!]
\centering
\includegraphics[width=0.50\textwidth]{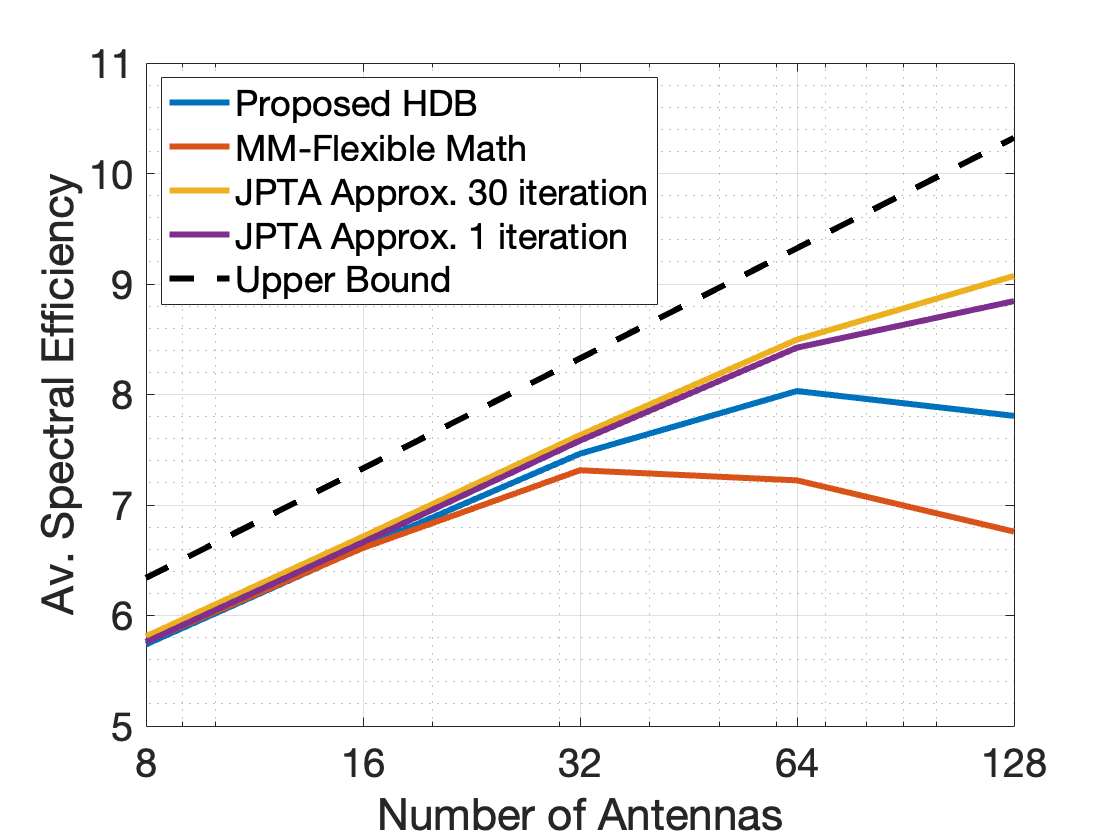}
\caption{\gls{ASE} of the proposed \gls{HDB} algorithm and benchmark algorithms for different number of antennas $N$ for $G=3$ \gls{UE}s and $BW=3$ GHz bandwidth. The \textit{Upper Bound} refers to the \gls{ASE} achieved with the maximum beamforming gain.}
\label{Fig: Ant 3 users}
\end{figure}

\begin{figure}[t!]
\centering
\includegraphics[width=0.50\textwidth]{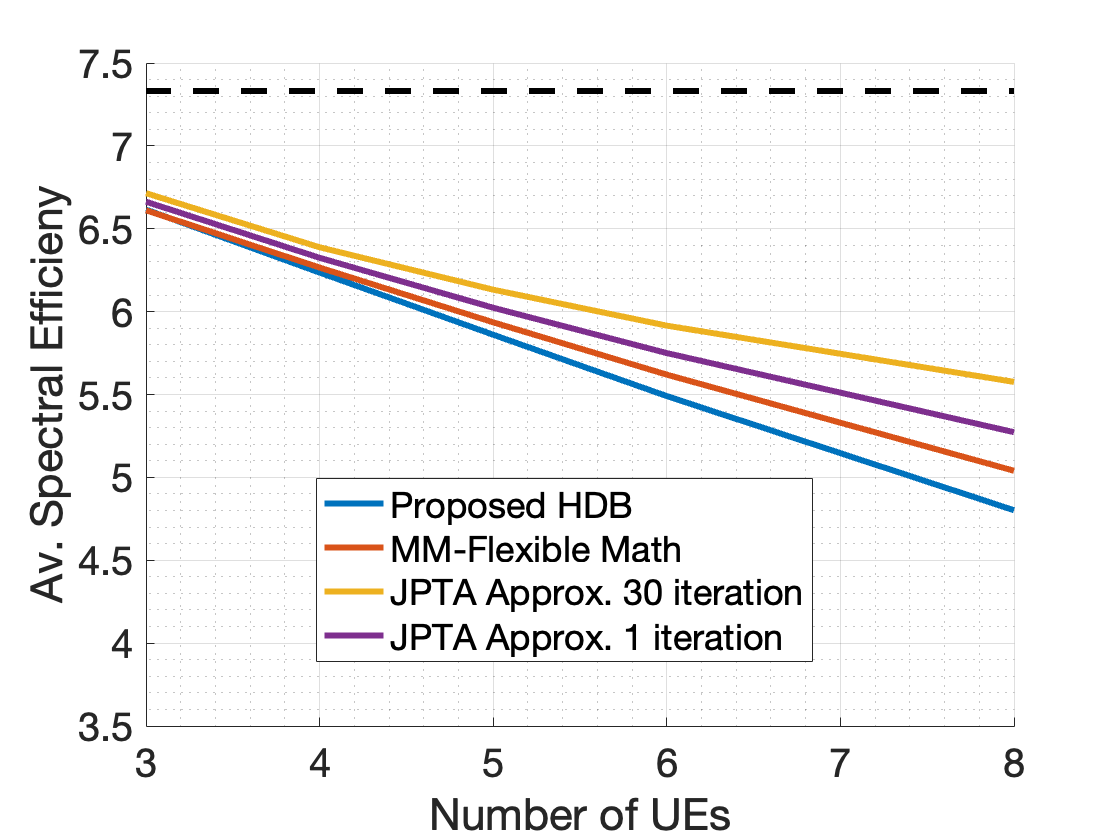}
\caption{\gls{ASE} of the proposed \gls{HDB} algorithm and benchmark algorithms for different number of \gls{UE}s $G$ for $N=16$ antennas and $BW=3$ GHz bandwidth. The \textit{Upper Bound} refers to the \gls{ASE} achieved with the maximum beamforming gain.}
\label{Fig: Effect of UE}
\end{figure} 

Figure \eqref{Fig: BW 3 users} shows the \gls{ASE} over different $BW$ values for $G=3$ \gls{UE}s and $N=16$ antennas. We can observe that the JPTA approximation algorithm \cite{jpta} performs the best across all bandwidth values and is minimally affected by the increased bandwidth. While the mmFlexible algorithm's \cite{mmflexible} performance degrades significantly with increased bandwidth, reducing from $90.98\%$ of the maximum \gls{ASE} at $1$ GHz to $78.16\%$ of the maximum \gls{ASE} at $10$ GHz bandwidth as a result of the beamsquint, which increases with bandwidth. \gls{HDB} algorithm performs almost identically to the benchmark algorithms for $BW<3$ GHz and gets less affected by the increased bandwidth compared to the mmFlexible algorithm \cite{mmflexible}. This reduction results from the accumulation of generator errors, where the JPTA approximation \cite{jpta} algorithm is utilized for generator dictionary creation. Results indicate that \gls{HDB} algorithm outperforms mmFlexible math \cite{mmflexible} algorithm for $BW>3$ GHz and is less affected by the increased bandwidth. 

Figure \eqref{Fig: Ant 3 users} demonstrates the \gls{ASE} over different numbers of antennas for $G=3$ \gls{UE}s and $BW=3$ GHz bandwidth. \gls{HDB} algorithm and benchmark algorithms perform almost identical for $N\leq 16$ antennas. We observe that the performance of the JPTA approximation algorithm \cite{jpta} scales almost linearly with the number of antennas $N$, deviating slightly for $N \geq 64$. In contrast, mmFlexible math algorithm \cite{mmflexible} and \gls{HDB} algorithms plateau and degrade after $N=32$ and $N=64$ antennas, respectively. As the number of antennas increases, the corresponding beamwidth decreases, and algorithms become more sensitive to errors, i.e., beam squint as in mmFlexible \cite{mmflexible} and generator error accumulation as in \gls{HDB} algorithm. We observe that \gls{HDB} algorithm outperforms the mmFlexible math algorithm for $N>32$ antennas.

Figure \eqref{Fig: Effect of UE} shows the \gls{ASE} over different numbers of \gls{UE}s for $N=16$ antennas and $BW=3$ GHz bandwidth. We observe that \gls{HDB} and benchmark algorithms perform similarly for $G\leq 4$ \gls{UE}s and \gls{ASE} decreases with the number of \gls{UE}s for all algorithms. Results indicate that utiliziation of split beampatterns is more suitable for small number of \gls{UE} deployments.

\section{Discussion and Future Directions} \label{Section: Discussion and Future Directions}

The \gls{HDB} algorithm is a low complexity and low-memory algorithm for split beampattern synthesis that is suitable for practical deployment. We showed that the proposed \gls{HDB} algorithm is less affected by the increased bandwidth and number of antennas compared to the mmFlexible algorithm \cite{mmflexible} while having similar computational complexity and negligible memory requirements. In addition, we empirically observe that \gls{HDB} algorithm provides higher minimum gain and fairness among users due to the post-processing of the generator beampattern dictionary discussed in Section \eqref{Section:Generator Dictionaries}. This property not only makes the \gls{HDB} more attractive to the systems that are aiming for fair service across \gls{UE}s but can also reduce the complexity of the upper layer scheduling and assignment algorithms.

In the future, we plan to optimize the shift and scale amount in equation \eqref{Eqn: fc and BW of generators} for various objective functions over the desired \gls{UE} configurations, adapting the \gls{HDB} algorithm for different deployment scenarios. The optimization of the generator beampattern dictionary for different objective functions is also an open problem. Lastly, we believe that the investigation of the extension of the observed mathematical structure to hybrid \gls{TTD} arrays is an interesting research direction.

\section{Conclusion} \label{Section: Conclusion}

In this article, we present a low-complexity and low-memory split beampattern approximation algorithm utilizing the mathematical structure of the beampattern synthesis operation. We first demonstrate the homomorphism between array configuration matrices and corresponding \gls{TTD} beampatterns. Then, we represent the desired hard-to-approximate split beampattern in terms of easy-to-approximate generator beampatterns. Finally, by utilizing a single generator beampattern dictionary, we approximate the desired split beampattern with the observed homomorphism in a divide-and-conquer manner. The proposed algorithm achieves comparable performance to the high-complexity or high-memory JPTA algorithms \cite{jpta} with orders of magnitude smaller memory requirements, without ignoring the beam squint as mmFlexible algorithm \cite{mmflexible}. Furthermore, with extensive simulations, we show that the proposed algorithm achieves higher minimum power and fairness across different \gls{UE}s compared to benchmark algorithms.

\section*{Acknowledgment}

The authors would like to thank Aditya Wadaskar and Ding Zhao for helpful early discussions.

\bibliography{references}
\bibliographystyle{ieeetr}  
\end{document}